\documentclass[onecolumn,secnumarabic,nobibnotes,superscriptaddress,aps]{revtex4-2}
\usepackage{color, xcolor, colortbl}
\usepackage{graphicx}
\usepackage{geometry}
\usepackage{amsthm,amsmath,amssymb,amsfonts}
\usepackage{dcolumn}
\usepackage{url}
\usepackage{epstopdf}
\usepackage{enumitem}
\usepackage{algorithm}
\usepackage{algpseudocode}
\usepackage{bm}
\usepackage[caption=false]{subfig}
\usepackage{appendix}
\usepackage{multirow}
\usepackage{braket}
\usepackage[english]{babel}
\usepackage{hyperref}
\usepackage[capitalize]{cleveref}
\usepackage{xpatch}
\usepackage{tikz}
\usepackage{adjustbox}
\usepackage{xspace}
\usetikzlibrary{quantikz}

\DeclareMathOperator*{\argmax}{arg\,max}

\newcommand{\I}{i}

\newcommand{\mc}[1]{\mathcal{#1}}

\newcommand{\wt}[1]{\widetilde{#1}}

\newcommand{\innerp}[2]{\left\langle #1 \vert #2 \right\rangle}
\newcommand{\abs}[1]{\left\lvert#1\right\rvert}
\newcommand{\norm}[1]{\left\lVert#1\right\rVert}

\newcommand{\rd}{\,\mathrm{d}}
\newcommand{\Or}{\mathcal{O}}

\newcommand{\ZZ}{\mathbb{Z}}

\newtheorem{thm}{\protect\theoremname}
\newtheorem{lem}[thm]{\protect\lemmaname}

\newtheorem{prop}[thm]{\protect\propositionname}

\newtheorem{prob}{\protect\problemname}

\providecommand{\definitionname}{Definition}
\providecommand{\assumptionname}{Assumption}
\providecommand{\corollaryname}{Corollary}
\providecommand{\lemmaname}{Lemma}
\providecommand{\propositionname}{Proposition}
\providecommand{\remarkname}{Remark}
\providecommand{\theoremname}{Theorem}
\providecommand{\problemname}{Problem}
\usepackage{bbm}

\makeatletter
\newenvironment{breakablealgorithm}
  {
   \begin{center}
     \refstepcounter{algorithm}
     \hrule height.8pt depth0pt \kern2pt
     \renewcommand{\caption}[2][\relax]{
       {\raggedright\textbf{\fname@algorithm~\thealgorithm} ##2\par}%
       \ifx\relax##1\relax 
         \addcontentsline{loa}{algorithm}{\protect\numberline{\thealgorithm}##2}%
       \else 
         \addcontentsline{loa}{algorithm}{\protect\numberline{\thealgorithm}##1}%
       \fi
       \kern2pt\hrule\kern2pt
     }
  }{
     \kern2pt\hrule\relax
   \end{center}
  }
\makeatother

\usetikzlibrary{fit}
\tikzset{%
  highlight/.style={rectangle,rounded corners,fill=blue!15,draw,fill opacity=0.3,thick,inner sep=0pt}
}

\begin{document}

\title{Robust ground-state energy estimation under depolarizing noise}

\newcommand{\DeptMath}{Department of Mathematics, University of California, Berkeley, CA 94720, USA}
\newcommand{\LBLMath}{Applied Mathematics and Computational Research Division, Lawrence Berkeley National Laboratory, Berkeley, CA 94720, USA}
\newcommand{\Caltech}{
Institute for Quantum Information and Matter, California Institute of Technology, Pasadena, CA 91125, USA}
\newcommand{\CIQC}{Challenge Institute of Quantum Computation, University of California, Berkeley, CA 94720, USA}

\author{Zhiyan Ding} 
\thanks{These authors contributed equally.}
\affiliation{\DeptMath}
\author{Yulong Dong} 
\thanks{These authors contributed equally.}
\affiliation{\DeptMath}
\author{Yu Tong}
\affiliation{\Caltech}
\author{Lin Lin}
\email{Electronic address: linlin@math.berkeley.edu}
\affiliation{\DeptMath}
\affiliation{\LBLMath}
\affiliation{\CIQC}

\begin{abstract}
We present a novel ground-state energy estimation algorithm that is robust under global depolarizing error channels. Building upon the recently developed Quantum Exponential Least Squares (QCELS) algorithm, our new approach incorporates significant advancements to ensure robust estimation while maintaining a polynomial cost in precision. By leveraging the spectral gap of the Hamiltonian effectively, our algorithm overcomes limitations observed in previous methods like quantum phase estimation (QPE) and robust phase estimation (RPE). Going beyond global depolarizing error channels, our work underscores the significance and practical advantages of utilizing randomized compiling techniques to tailor quantum noise towards depolarizing error channels. Our research demonstrates the feasibility of ground-state energy estimation in the presence of depolarizing noise, offering potential advancements in error correction and algorithmic-level error mitigation for quantum algorithms.
\end{abstract}
\maketitle
\section{Introduction}

Estimating the ground-state energy of a quantum Hamiltonian is a fundamentally important task in many areas of quantum science, including quantum chemistry, condensed matter physics, and quantum complexity theory. This task has been considered a promising candidate for practical quantum advantage, in which quantum computers can help achieve significant speedup, enabling efficient computation for quantum systems that are difficult for classical computers to handle \cite{LeeLeeZhaiEtAl2022there,CaoRomeroEtAl2019quantum,BauerBravyiEtAl2020quantum,McArdleEndoEtAl2020quantum}.
On fault-tolerant quantum computers with sufficiently low logical error rate achieved through quantum error correction, of which impressive experimental progress has been made recently \cite{SivakEickbuschEtAl2022real,GoogleQuantumAI2023suppress}, quantum phase estimation (QPE) \cite{Kitaev1995,AbramsLloyd1999} based algorithm can possibly solve tasks that are difficult on classical computers \cite{BeverlandMuraliEtAl2022assessing}.
Many variants of the QPE algorithm and post-QPE algorithms have been developed for the ground-state energy estimation problem with improved runtime scaling \cite{PoulinWocjan2009,GeTuraCirac2019,LinTong2020a}. Furthermore, significant progress has been made in recent years \cite{PhysRevA.104.069901,Campbell2021EarlyFS,LinTong2022,Russo2021EvaluatingED,ZhangWangJohnson2022computing,WangSimJohnson2022state,ZengSunYuan2021universal,DongLinTong2022ground,DingLin2023,DingLin2023simultaneous,WangStilkFrancaEtAl2022quantum,LiNiYing2023low} in the development of ground-state energy estimation algorithms using very few ancilla qubits. These advancements enable a favorable trade-off between circuit depth, ancilla qubits, and total runtime, making them suitable for early fault-tolerant quantum computers while providing provable performance guarantees. These algorithms assume that we have access to a circuit that can prepare an initial state with significant overlap with the ground state. Possible choices for this initial state include the Hartree-Fock state, the matrix-product state from density-matrix renormalization group algorithm, and states obtained from variational quantum eigensolver (VQE) \cite{peruzzo2014variational,mcclean2016theory,omalley2016scalable}, and the quality of the initial state should be carefully examined in practice~\cite{LeeLeeZhaiEtAl2023}.

While the progress mentioned above can lighten the burden on error correction by reducing the size of the circuit, attempting to suppress errors through exhaustive error correction inevitably increases the amount of quantum resources needed as well as the runtime. It may therefore be reasonable to hope that quantum computers can generate sufficiently accurate quantum data, and then the impact of the quantum noise can be mitigated at the algorithmic level  on classical computers \cite{KshirsagarEtAl2022proving}. This requires analyzing the performance of quantum algorithms under specific error models, and possibly adjusting the algorithms to increase the error tolerance. For instance, a recent work along this line analyzed the performance of quantum signal processing \cite{LowChuang2017} with imprecise rotation angles in the signal operator, and designed a procedure to correct this coherent noise channel up to arbitrarily high precision \cite{TanLiuTranChuang2023error}. Another common approach for mitigating coherent noise is to employ techniques like randomized compiling \cite{WallmanEmerson2016} to transform the noise into stochastic Pauli errors. Although physical noise typically exhibits local characteristics, it is possible for quantum circuits to spread local errors and manifest them as global depolarizing noise \cite{BoixoIsakovSmelyanskiyEtAl2018,DalzellHunterJonesBrandao2021random, FoldagerKoczor2023can}. The error results in an exponential decay in the observable expectation values. Therefore, conducting thorough evaluations of quantum algorithm performance under the global depolarizing noise model can yield significant theoretical and experimental insights.
Although simple error mitigation techniques can be used to mitigate its effect, the total cost of generic error mitigation schemes may require a large number of measurements, and the total cost may scale exponentially with respect to the precision~\cite{TakagiTajimaEtAl2022universal,QuekFrancaEtAl2022exponentially}.

\subsection{Problem setup}
\label{sec:setup_of_the_problem}

Given a quantum Hamiltonian $H\in\mathbb{C}^{M\times M}$ with eigenvalue/vector pairs $\left\{(\lambda_m,\ket{\psi_m})\right\}^M_{m=1}$, our objective is to estimate the ground state energy $\lambda_0$ of $H$. For simplicity, we assume that $\lambda_0\in(-\pi,\pi]$, and the spectral gap $\Delta=\lambda_1-\lambda_0>0$, and $H$ is normalized so that $\|H\|\leq 1$. We also assume access to a sufficiently good (but not perfect) initial quantum state $\ket{\psi}$ with $p_0=\left|\innerp{\psi_0}{\psi}\right|^2>0.5$, and a family of noisy Hamiltonian evolution operator $\exp(-itH)$. Specifically, we assume that the circuit depth for implementing the controlled time evolution $U_t=\ket{0}\left\langle 0\right|\otimes I_M+\ket{1}\left\langle 1\right|\otimes \exp(-itH)$ grows linearly in $t$ (in other words, the Hamiltonian cannot be fast-forwarded), and the global depolarizing noise channel transforms $U_t$ into a quantum channel $\mathcal{M}_t$:  \begin{equation}\label{eqn:M_t}
\mathcal{M}_t(\rho)=\exp(-\alpha |t|)U_t\rho U^\dagger_t+\frac{\left(1-\exp(-\alpha |t|)\right)}{2M} I_{2M}\,.
\end{equation}
Here, $\rho\in\mathbb{C}^{2M\times 2M}$ is an arbitrary density operator of the system and ancilla register, $\alpha>0$ is the parameter that characterizes the strength of the noise, and $I_{2M}$ is the identity matrix of size $2M$. Intuitively, 
 because $\left\|\mathcal{M}_t(\rho)-\frac{1}{2M}I_{2M}\right\|_1\leq \exp(-\alpha|t|)\ll 1$ ($\|\cdot\|_1$ is the trace distance), when $t\gg \alpha^{-1}$, extracting accurate information from $U_t$ and $\rho$ using $\mathcal{M}_t(\rho)$ would be challenging when $t$ greatly exceeds $\alpha$.

This problem considered in this paper is as follows:
 
\begin{prob}\label{prob:main}
Given a gapped Hamiltonian $H$, a sufficiently good initial state $\ket{\psi}$, and a family of noisy Hamiltonian evolution channel $\mc{M}_t$ in \cref{eqn:M_t}, is it possible to estimate $\lambda_0$ to \textbf{any} precision $\epsilon>0$, and the total cost scales \textbf{polynomially} in $1/\epsilon$?
\end{prob}

To see why \cref{prob:main} is challenging, we first note that in order to estimate $\lambda_0$ to precision $\epsilon$, the maximal runtime of QPE, and many other methods such as  robust phase estimation (RPE) \cite{PhysRevA.104.069901} needs to be $\Omega(\epsilon^{-1})$. QPE is based on the quantum Fourier transform, which requires very precise implementation of the Hamiltonian evolution. This mechanism completely breaks down when $\epsilon$ is comparable to, or smaller than the noise rate $\alpha$ (for a concrete demonstration, see numerical results in \cref{sec:ins}). Furthermore, the performance of QPE cannot be improved by reducing the Monte Carlo sampling error. RPE has the salient property that in the limit of zero sampling error, the result is invariant in the presence of the global depolarizing noise channel. However, as analyzed in Appendix \ref{sec:intuitive_hardness}, when the sampling error is taken into account and $\epsilon\ll \alpha$, the number of measurements (and hence the total cost) of RPE can scale \textit{exponentially} in $1/\epsilon$. One major challenge faced by methods like QPE and RPE is their inability to leverage the spectral gap $\Delta$ to effectively decrease the maximum runtime.
%
%
%

\section{Results}

 We provide a positive answer to \cref{prob:main} by presenting an algorithm that achieves a maximal runtime of $\wt{\Theta}(\log(1/\epsilon)\Delta^{-1})$ and a total runtime of $\wt{\Theta}((1/\epsilon)^{\alpha/\Delta})$  for any $\epsilon>0$. 
This implies that the algorithm can estimate $\lambda_0$ to any $\epsilon$-accuracy, while the total runtime scales polynomially in $1/\epsilon$, given that the depolarizing noise constant $\alpha$ is not significantly greater than the gap $\Delta$. Our algorithm is based on the recently developed quantum complex exponential least squares (QCELS) method \cite{DingLin2023, DingLin2023simultaneous}. Unlike QPE or RPE, the combination of QCELS and a spectral filter proposed in~\cite{LinTong2022} can exploit the existence of a spectral gap to decrease the maximum runtime to $\wt{\Theta}(\Delta^{-1})$~\cite[Corollary 4]{DingLin2023}. However, the robustness of the spectral filtering step in the presence of the depolarizing error channel is unclear. One key theoretical advancement and algorithmic improvement in this paper is that, given a sufficiently good initial state, we may modify QCELS to directly estimate the ground-state energy to any precision under the global depolarizing noise, and the maximal runtime is only $\wt{\Theta}(\log(1/\epsilon)\Delta^{-1})=\Or(\alpha^{-1})$. This bypasses the aforementioned difficulty of QPE, RPE, and the original version of QCELS. The requirement for $T_{\max}$ to be greater than $\Delta^{-1}$ is indeed expected.  Essentially, to estimate $\lambda_0$ with high accuracy $\epsilon\ll \Delta$, it is crucial to distinguish between the signals $\exp(-i\lambda_0 t)$ and $\exp(-i\lambda_1 t)$. Therefore, it is necessary to choose a value of $t$ that is larger than $\frac{1}{\lambda_1-\lambda_0}$ to that the difference between these two signals is on the order of $\mathcal{O}(1)$.

We observe that in practical applications, $\norm{H}$ can exhibit polynomial growth with respect to the number of qubits $n$ before normalization. As a result, upon renormalization, the gap $\Delta$ scales as $\mathrm{poly}(n^{-1})$. In such situations, our theoretical findings indicate that the noise constant $\alpha$ must also decrease following a polynomial trend of $\mathrm{poly}(n^{-1})$ to ensure the ratio $\alpha/\Delta$ remains constant. This requirement is stringent and might not be easily achievable.

Our numerical results obtained under the assumption of the global depolarizing noise channel provide compelling evidence for the superior performance of our new algorithm. However, it is important to acknowledge that the noise model employed in real quantum devices is significantly more complex and extends beyond the scope of the global depolarizing error channel. In order to gain a better understanding of how our algorithm performs on real devices, we take an initial step by conducting additional simulations utilizing various realistic local channels, such as coherent noise, bit flip, phase flip, and general Pauli error channels.
These simulations serve to demonstrate the robustness of our algorithm not only against the global depolarizing channel but also under the influence of local depolarizing channels. The results obtained from these tests highlight the significance of randomized compiling in enhancing the numerical outcomes of our algorithm. This information is valuable in improving our comprehension of the algorithm's behavior in realistic scenarios and sets the stage for further investigations into its performance on real quantum hardware.


%

\subsection{Algorithm}\label{sec:main_algorithm}

\begin{figure}[htbp]
    \centering
    \includegraphics[width=\textwidth]{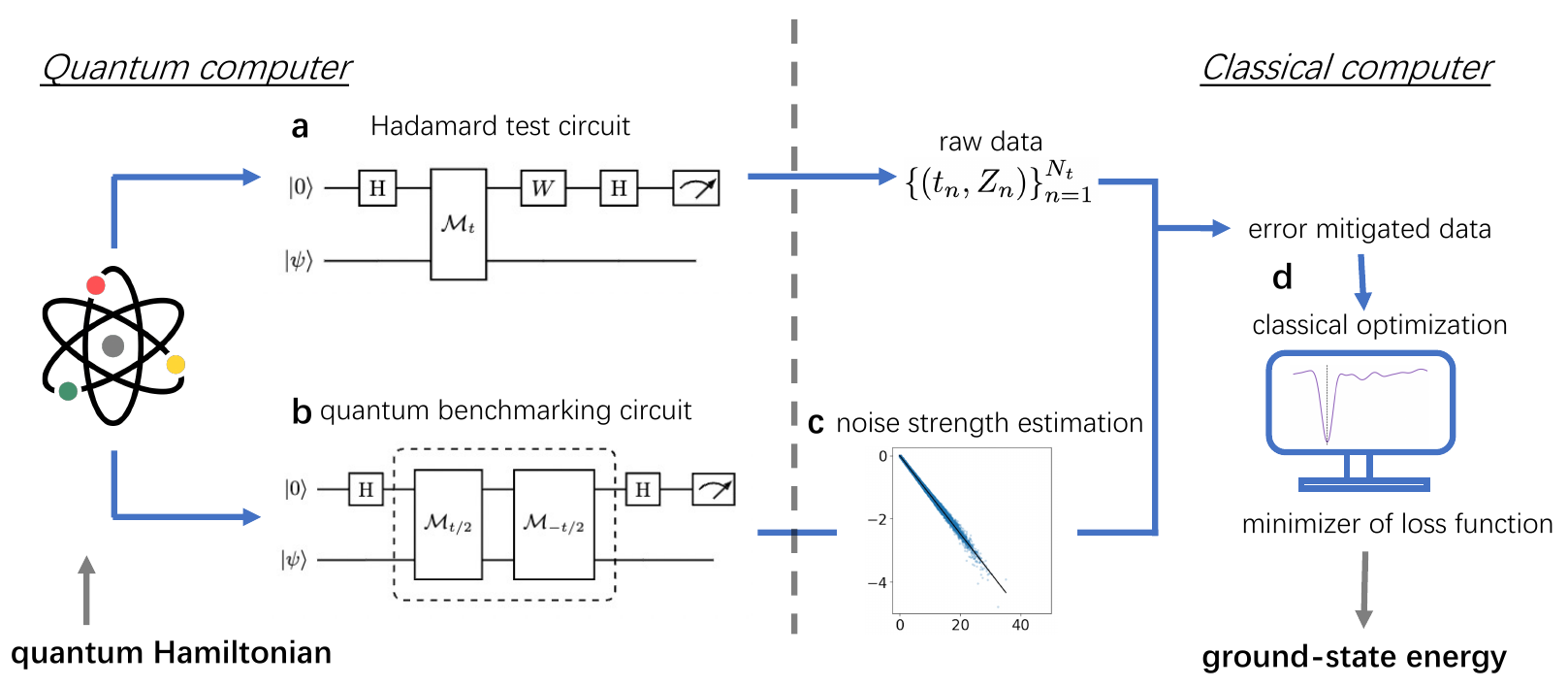}
    \caption{Flowchart of the main algorithm. \textbf{a} The quantum circuit for Hadamard test. \textbf{b} The quantum circuit for quantum benchmarking. \textbf{c} Performing regression to estimate the noise strength based on the output of quantum benchmarking circuits. \textbf{d} Classically computing the error mitigated data and estimating the ground-state energy of the Hamiltonian by classically minimizing a loss function. }
    \label{fig:flowchart}
\end{figure}

\begin{breakablealgorithm}
      \caption{Ground-state energy estimation algorithm resilient to  global depolarizing noise}
  \label{alg:main}
  \begin{algorithmic}[1]
  \Statex \textbf{Input:} Number of time samples: $N_{t}$; Number of data samples: $N_{s,1},N_{s,2}$; Parameters: $\gamma,T,a$; Times : $\{t_{B,n}\}^{N_B}_{n=1}$
    \State Generate $\widetilde{\alpha}$ using Algorithm \ref{alg:estimation_alpha} with $\{t_{B,n}\}^{N_B}_{n=1},N_{s,1}$.
    
   \State Generate a data set $\left\{(t_{n},Z_{n})\right\}^{N_t}_{n=1}$ using Algorithm \ref{alg:data} with the probability density $a(t)$ and $N_{s,2}$ data samples for each $t_n$.

   \State $(r^*,\theta^*)\gets\mathrm{argmin}_{r\in\mathbb{C},\theta\in[-\pi,\pi]}\frac{1}{N_t}\sum^{N_t}_{n=1}\left|\exp(\widetilde{\alpha} |t_n|)Z_n-r\exp(-i\theta t_n)\right|^2$

    \Statex \textbf{Output:} An estimation of $\lambda_0$: $\theta^*$
    \end{algorithmic}
\end{breakablealgorithm}

Our algorithm contains three steps:
\begin{enumerate}
 
\item Estimate the parameter noise $\alpha$ (on the quantum computer): Our method to estimate $\alpha$ is based on the following equality:
\begin{equation}\label{eqn:Bt}
\mc{B}_t(\rho):=\mathcal{M}_{-t/2}\circ\mathcal{M}_{t/2}(\rho)=\exp(-\alpha |t|)\rho+\frac{\left(1-\exp(-\alpha |t|)\right)}{2M} I_{2M}.
\end{equation}
From this observation, for each $t_{B,n}$, we can employ the Hadamard test with $\mathcal{B}_{t_{B,n}}$ (see \cref{fig:flowchart} for the circuit) to obtain a random variable $B_n$ such that $\mathbb{E}(B_n)=\exp(-\alpha|t_{B,n}|)$.  Then, we use the linear regression to fit the set of data pairs $\{(|t_{B,n}|,-\log(B_n))\}^{N_B}_{n=1}$ and define the approximation $\widetilde{\alpha}$ as the slope of the linear fitting curve.

A more detailed description of the quantum circuit and the construction algorithm are referred to Appendix \ref{sec:estimation_beta} respectively.

\item Data generation (on the quantum computer): We first draw a set of $N_t$ time points $\{t_n\}_{n=1}^{N_t}$ independently from a truncated normal distribution $a(t)$, where $a(t)$ is a truncated probability density function of the normal distribution $\mathcal{N}(0,T)$, with support in $[-\gamma T,\gamma T]$ and parameter $\gamma>0$. We then use these time points to perform the Hadamard test with quantum channel $\mathcal{M}_{t_n}$ \eqref{eqn:M_t} and $W=I$ or $S^\dagger$, where $S$ is the phase gate. We measure the ancilla qubit $N_{s,2}$ times at each time point $t_n$ to prepare a data set $\mathcal{D}_H = \{(t_n, Z_n)\}_{n=1}^{N_t}$, where $Z_n$ is a random variable satisfying $
        \mathbb{E}(Z_n) = S(t_n):=\exp(-\alpha|t_n|)\sum^{M-1}_{m=0}p_m\exp(-i \lambda_m t_n)$.

    The detailed algorithm for the data generation process is presented in Appendix \ref{sec:mqc} \cref{alg:data}.

\item Optimization (on the classical computer): 
\begin{equation}\label{eqn:op}
\left(r^*,\theta^*\right)=\mathrm{argmin}_{r\in\mathbb{R},\theta\in[-\pi,\pi]}L\left(r,\theta\right)\,,
\end{equation}
with the loss function
\begin{equation}\label{eqn:fixed_beta_new}
L\left(r,\theta\right)=\frac{1}{N_t}\sum^{N_t}_{n=1}\left|\exp\left(\widetilde{\alpha} |t_n|\right)Z_n-r\exp(-i\theta t_n)\right|^2.
\end{equation}
By minimizing this loss function $L\left(r,\theta\right)$, we can obtain an estimation $\widetilde{\lambda}_0=\theta^*$.

\end{enumerate}


A flowchart of our main algorithm can be found in \cref{fig:flowchart}.

In the noise estimation step, we note that $\mathcal{B}_t$ can be implemented by propagating the noisy controlled Hamiltonian simulation both forward and backward in time. When there is no noise present, $\mathcal{B}_t$ reduces to the identity map. We anticipate that the circuit structure for implementing $\mathcal{B}_t$ and $\mathcal{M}_t$ would be quite similar, differing primarily in their parameterization. Hence, we expect the estimated value of $\alpha$ obtained from Eq. (2) to exhibit a strong correlation with the circuit used for generating the quantum data. This procedure shares similarities with the Loschmidt echo~\cite{CucchiettiDalvitPazEtAl2003}  and the mirror benchmarking procedure  \cite{ProctorSeritanRudingerEtAl2021}.

For a given set of time steps $\{t_{B,n}\}^{N_B}_{n=1}$, we construct Hadamard test circuits that correspond to the quantum channels $\mathcal{B}_{t_{B,n}}$. By performing measurements on the ancilla qubit of each circuit $N_{s,1}$ times, we acquire a set of variables $\{B_{n}\}^{N_B}_{n=1}$, where each $B_{n}$ satisfies $\mathbb{E}(B_{n})=\exp(-\alpha |t_{B,n}|)$ and $|B_{n}-\exp(-\alpha |t_{B,n}|)|=\mathcal{O}(1/\sqrt{N_{s,1}})$ with high probability. Although our algorithm does not explicitly depend on knowledge of the spectral gap $\Delta$, it can be beneficial to compare $\widetilde{\alpha}$ with an approximate value of the gap. If the strength of the noise significantly surpasses the magnitude of the gap, we might need to substantially increase the number of samples to achieve accurate estimates of the ground-state energy.

Finally, in contrast to the original QCELS method described in \cite{DingLin2023}, our algorithm adopts a different strategy in the second step. Instead of using uniformly distributed time points, we sample random time points from a Gaussian distribution. This choice of a Gaussian distribution plays a vital role in the success of our algorithm. It is worth highlighting that the Fourier transform of a Gaussian distribution is also a Gaussian function, which exhibits a high concentration around the origin. This property makes it an effective filter function for isolating $\lambda_0$ and filtering out other eigenvalues. The detailed explanation and intuitive derivation can be found in \cref{sec:complexity_algorithm}.

\subsection{Complexity analysis}\label{sec:complexity_result}

In the optimization step, when both $N_{s,2}$ and $N_t$ are sufficiently large, the loss function $L\left(r,\theta\right)$ closely approximates the ``ideal'' loss function:
\begin{equation}\label{eqn:loss_multi_modal_perfect}
\mathcal{L}\left(r,\theta\right)=\int^\infty_{-\infty}a(t)\left|\exp\left((\widetilde{\alpha}-\alpha) |t|\right)\sum^{M-1}_{m=0}p_m\exp(-i \lambda_m t)-r\exp(-i\theta t)\right|^2\rd t,,
\end{equation}
and \eqref{eqn:op} is equivalent to
\begin{equation}\label{eqn:op_perfect}
\left(r^*,\theta^*\right)=\mathrm{argmin}_{r\in\mathbb{C},\theta\in[-\pi,\pi]}\mathcal{L}\left(r,\theta\right)\,.
\end{equation}
Intuitively, when $p_0>0.5$, $\widetilde{\alpha}\approx \alpha$, to achieve the minimum, we should have $\theta^*\approx \lambda_0$ so that $r^*\exp(-i\theta^* t)$ can cancel the dominant term $p_0\exp(-i\lambda_0 t)$ in the original signal. The QCELS method \cite{DingLin2023} is based on the same intuition, but the original proof requires the maximal simulation time to be $\Or(\epsilon^{-1})$ which is impractical in the noisy setting. Furthermore, in Appendix \ref{sec:analysis_QCELS}, we provide a simple example to illustrate that, under the global depolarizing noise, the error of QCELS has a positive lower bound for approximating the ground state energy no matter how large $T_{\max}$ is. In contrast, our analysis in \cref{thm_alg_new_informal} shows that Algorithm \ref{alg:main} can significantly reduce the requirement of the maximal runtime and the optimization problem with the new loss function in \cref{eqn:fixed_beta_new} can indeed estimate $\lambda_0$ to arbitrarily high precision in polynomial time. 

Our informal complexity result can be summarized as follows: 
\begin{thm}\label{thm_alg_new_informal}
Assume $p_0>0.5$. Given any $0<\epsilon<1/2$, there exists a sequence of time points $\{t_{B,n}\}^{N_B}_{n=1}$ such that $t_{B,n}=\mathcal{O}(1/\alpha)$ and we can choose $N_{s,1}=\widetilde{\Theta}(1/\epsilon^2)$ 
to obtain 
$|\widetilde{\alpha}-\alpha|\leq \epsilon$ with high probability. 

Next, we can choose $T=\Theta\left(\log^{1/2}\left(1/\epsilon\right)/\Delta\right)$ and $\gamma= \Theta\left(\log^{1/2}(1/\epsilon)\right)$ with $N_{s,2}=\Theta(1)$ and $N_t=\widetilde{\Theta}\left(\left(1/\epsilon\right)^{2+\Theta\left(\alpha/\Delta\right)}\right)$to guarantee that with high probability, $|\widetilde{\lambda}_0-\lambda_0|\leq \epsilon$. 

In summary, the maximal runtime of our algorithm is $T_{\max}=\Theta\left(\log\left(1/\epsilon\right)/\Delta\right)$ 
and the total runtime is
$
T_{\mathrm{total}}=\widetilde{\Theta}\left(\left(1/\epsilon\right)^{2+\Theta\left(\alpha/\Delta\right)}\right)
$.
\end{thm}


In \cref{sec:complexity_algorithm}, we discuss the detailed theorem and proof, which mainly consists of three parts. First, the complexity of estimating $\alpha$ can be derived directly from the linear regression formula. We note that the requirement $\widetilde{\alpha}\approx \alpha$ is important to get rid of the effect of global depolarizing noise in the loss function. Next, when we have enough data points and $\widetilde{\alpha}\approx \alpha$, solving \eqref{eqn:op} is approximately equivalent to maximizing $L(\theta)=\left|\int^\infty_{-\infty}a(t)\sum^{M-1}_{m=0}p_m\exp(i (\lambda_m-\theta) t)\rd t\right|^2$. Because $a(t)$ is an approximated Gaussian distribution, its Fourier transform also has the Gaussian form and concentrates around $0$. Thus, we can expect that $L(\theta)$ concentrates around every eigenvalue $\lambda_k$ and achieves its maximum value around $\lambda_0$ since $p_0>0.5$. Due to the concentration property of $F(\theta)$, it suffices to choose a maximal running time $T_{\max}=\Theta\left(\log\left(1/\epsilon\right)/\Delta\right)$ to eliminate the effect of $\lambda_1,\dots,\lambda_{M-1}$ in the optimization problem. Finally, since the maximal running time has a logarithmic dependence on $1/\epsilon$, the variance of $Z_n$ for each $t_n$ is bounded by $\exp(\alpha T_{\max})=\mathrm{poly}(1/\epsilon)$. This suggests that the necessary number of samples to accurately approximate the \emph{ideal} optimization problem ~\eqref{eqn:op_perfect} is of the order $\mathrm{poly}(1/\epsilon)$. As a result, the total running time scales polynomially in $1/\epsilon$.

\subsection{Numerical results: global depolarizing noise model}\label{sec:ins}


The ground-state energy estimation problem considered in the numerical experiments is that of the transverse-field Ising model (TFIM) whose Hamiltonian is defined as
\begin{equation}\label{eqn:H_Ising}
H=-\sum^{L-1}_{i=1} Z_{i}Z_{i+1} -g\sum^L_{i=1} X_i
\end{equation}
where $L$ is the number of qubits and $g$ is referred to as the coupling coefficient. We consider the case that $L = 4$ and $g = 1$. The initial state 
is prepared to $\ket{\psi} = \mathrm{H}^{\otimes L} \ket{0^L} = \ket{+^L}$, where $\mathrm{H}$ is the Hadamard gate and the state has overlap $p_0 \approx 0.8134$. In the numerical simulation, we normalize and shift the spectrum of the Hamiltonian $\widetilde{H}=H/\|H\|$ so that its eigenvalues lie in the interval $[-1, 1]$ which is independent of its scale. To demonstrate the performance with variable noise strength, the numerical experiments are performed with the noise strength $\alpha = 0.125, 0.25$.

We compare the performance of Algorithm \ref{alg:main}, QPE, RPE, and QCELS with global depolarizing noise. We derive the output distribution of QPE under global depolarizing noise and conduct experiments to subsequently sample from this newly obtained distribution. To fit the global depolarizing noise, the implementation of QCELS in our experiment is slight different from \cite{DingLin2023}. The version of QCELS used for comparison here can be seen as a special case of \cref{eqn:fixed_beta_new}, where we set $\wt{\alpha}=0$ and find the optimal $\theta$ in the entire complex plane  (see detail in Appendix \ref{sec:implemention_details_ideal}). We choose the parameters properly for RPE, QCELS, and Algorithm \ref{alg:main} to make sure the total running times of different methods are comparable. In addition, we simulate four methods ten times with random initial states $\{\ket{\psi^r_m}\}^{10}_{m=1}$ and measure the averaged estimation error defined as $|\widetilde{\lambda}_0-\lambda_0|$. Here, $\ket{\psi^r_m}$ is generated by randomly reordering the overlaps between the primary initial state $\ket{\psi}=\ket{+^L}$ and other eigenvalues. Specifically, $\ket{\psi^r_m}=\sum^2_{j=0}p_j\ket{\psi_j}+\sum^{2^L-1}_{j=3}p_{r_{j,m}}\ket{\psi_j}$, where $\{p_{r_{j,m}}\}^{2^L-1}_{j=3}$ is a random reordering set of $\{p_j\}^{2^L-1}_{j=3}$. The detailed implementation of different methods and the choices of parameters can be found in Appendix \ref{sec:implemention_details_ideal}.

The results shown in Figure \ref{fig:TFIM_loss_classical} indicate that QPE with depolarizing noise fails to accurately estimate $\lambda_0$, as the error does not decrease with increasing $T_{\max}$. The errors for RPE and QCELS decay linearly in $1/T_{\max}$ when $T_{\max}$ increases.  In contrast, Algorithm \ref{alg:main} exhibits a faster rate of error reduction with increasing $T_{\max}$ compared to RPE and QCELS, which is consistent with the theoretical result presented in Theorem \ref{thm_alg_new_informal}. We observe that the performance of Algorithm \ref{alg:main} surpasses our theoretical result stated in Theorem \ref{thm_alg_new_informal}. Specifically, the spectral gap in this case is approximately $\Delta \approx 0.25$. According to Theorem \ref{thm_alg_new_informal}, in order to achieve accuracy $\epsilon$, a choice of $N_t=\widetilde{\Theta}\left(\left(1/\epsilon\right)^{2+\Theta\left(\alpha/\Delta\right)}\right)\geq 10^6$ is required when $\epsilon\approx 10^{-3}$ and $\alpha=0.25$. However, in our experiment, we found that $N_t=10^4$ already yields an error of order almost $10^{-3}$. 


\begin{figure}[htbp]
     \subfloat{
         \centering
         \includegraphics[width=0.48\textwidth,height=0.3\textheight]{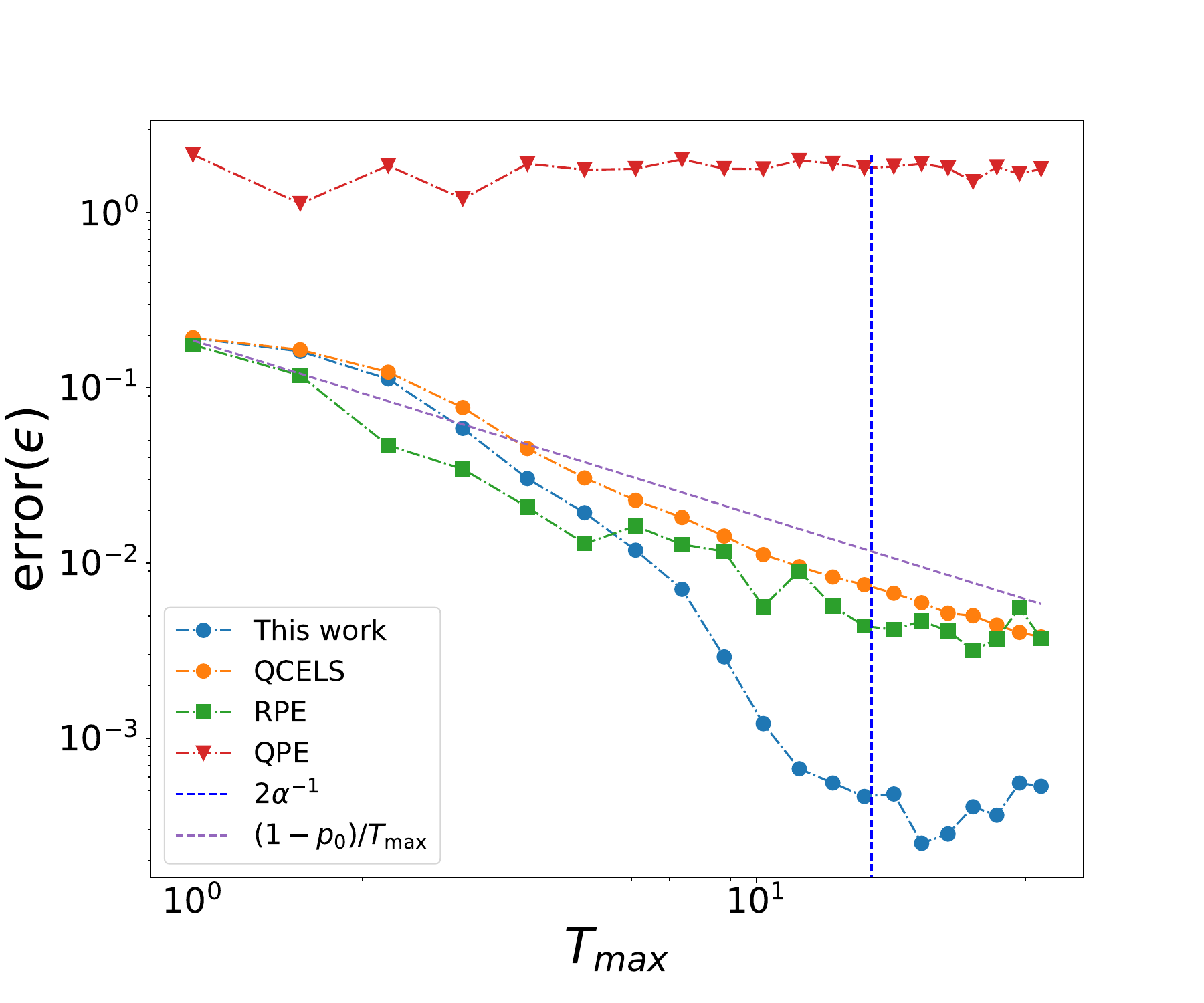}
     }
     \subfloat{
         \centering
         \includegraphics[width=0.48\textwidth,height=0.3\textheight]{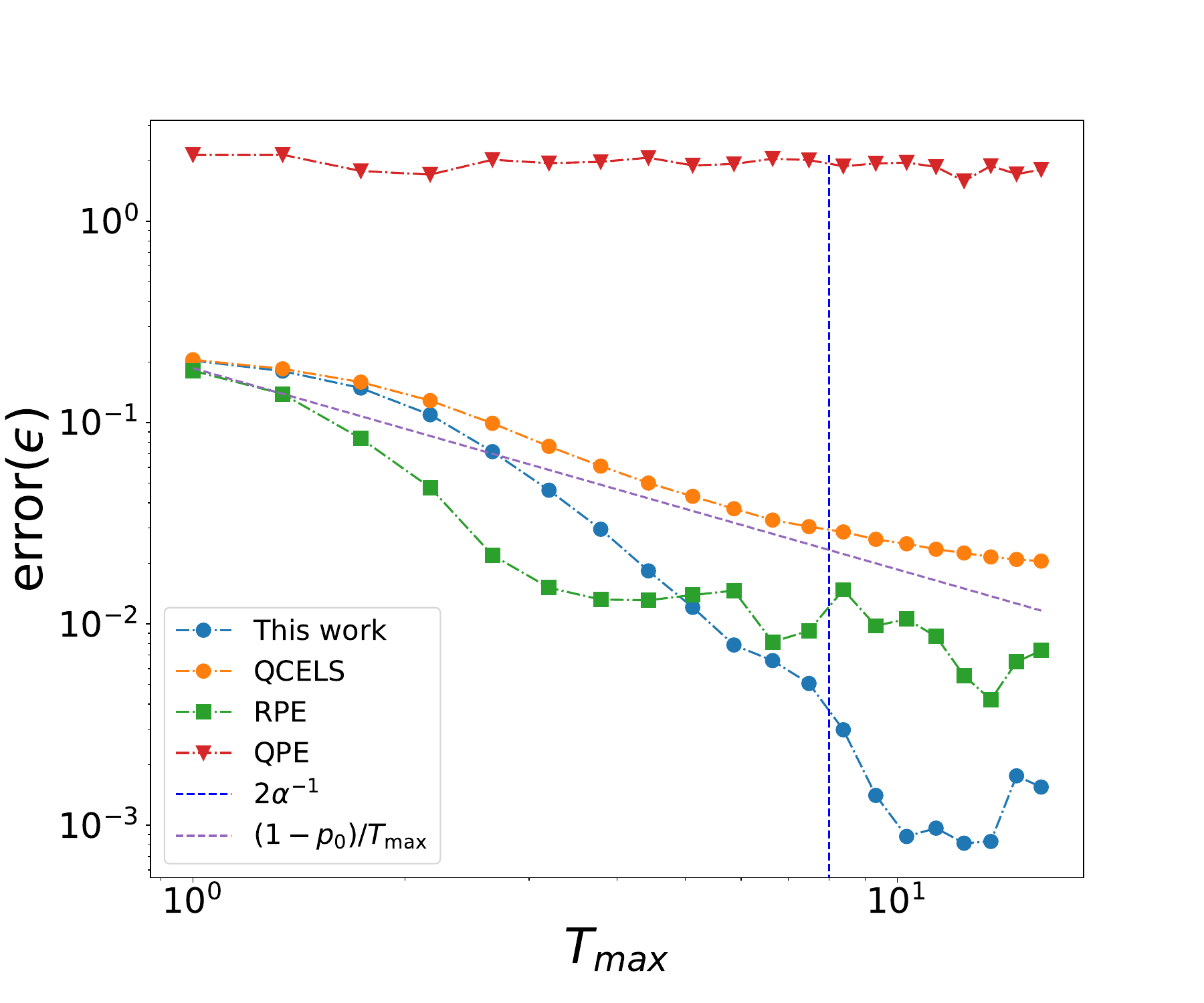}
     }
     \caption{
     \label{fig:TFIM_loss_classical} Comparison of Algorithm \ref{alg:main}, QPE, and RPE for the TFIM model with 4 sites and an initial overlap of $p_0=0.8$. Left: Noise parameter $\alpha=0.125$; Right: Noise parameter $\alpha=0.25$. 
     Our observations indicate that QPE fails to provide an accurate estimation for $\lambda_0$ in the presence of depolarizing noise, whereas Algorithm \ref{alg:main} performs much better than other methods. }
\end{figure}

\subsection{Numerical results: realistic local noise model}\label{sec:numerical_robust}

Although our algorithm is designed under the assumption of global depolarizing noise, we provide numerical evidence to demonstrate the robustness of our algorithm.  We perform noisy quantum experiments using the \textsf{IBM Qiskit} QASM simulator. The numerical tests involve multiple types of quantum noises, aligned with the noise strength parameter $\alpha$. The details of the simulation is given in \cref{sec:implemention_details_robust}.

The results in \cref{fig:tfim_regression} demonstrate the impact of increasing $T_\mathrm{max}$ on the accuracy of ground-state energy estimation. Although the locally applied quantum noises channels deviate from the ideal global noise ansatz, \cref{alg:main} still accurately estimates the ground-state energy. Incorporating quantum benchmarking circuits for error mitigation enhances the accuracy of energy estimation, particularly for small circuit depths.
However, the performance of \cref{alg:main} varies significantly
for deeper circuits with larger $T_\mathrm{max}$ under different noise models. \cref{fig:tfim_regression} shows that error mitigation through quantum benchmarking circuits significantly improves the accuracy of simulations with local depolarizing noise. The effect of employing the quantum benchmark circuit is clear even beyond $T_{\max}=2/\alpha$.
However, the use of the quantum benchmarking circuit is less effective beyond $T_{\max}=1/\alpha$ for other noise models, since these models deviate further from the global depolarizing assumption.

Among the various noise channels, coherent unitary noise stands out as a particularly channeling case. In \cref{sec:appendix_additional_numerical_results}, we delve into this behavior and demonstrate that coherent noise alters the effective Hamiltonian simulated in the quantum circuit. This discrepancy between coherent noise channels and Pauli channels underscores the significance of randomized compilation procedure. By employing Pauli twirling \cite{KnillLeibfriedReichleEtAl2008,WallmanEmerson2016}, an arbitrary quantum noise channel can be mapped to a local Pauli noise channel. Furthermore, incorporating twirling with the Clifford group enables further reduction of quantum noise to local depolarizing noise~\cite{MagesanGambettaEmerson2011}. Our findings indicate that leveraging this hierarchical noise reduction through twirling can yield significantly improved accuracy in energy estimation.

\begin{figure}[htbp]
    \centering
    \includegraphics[width = \textwidth]{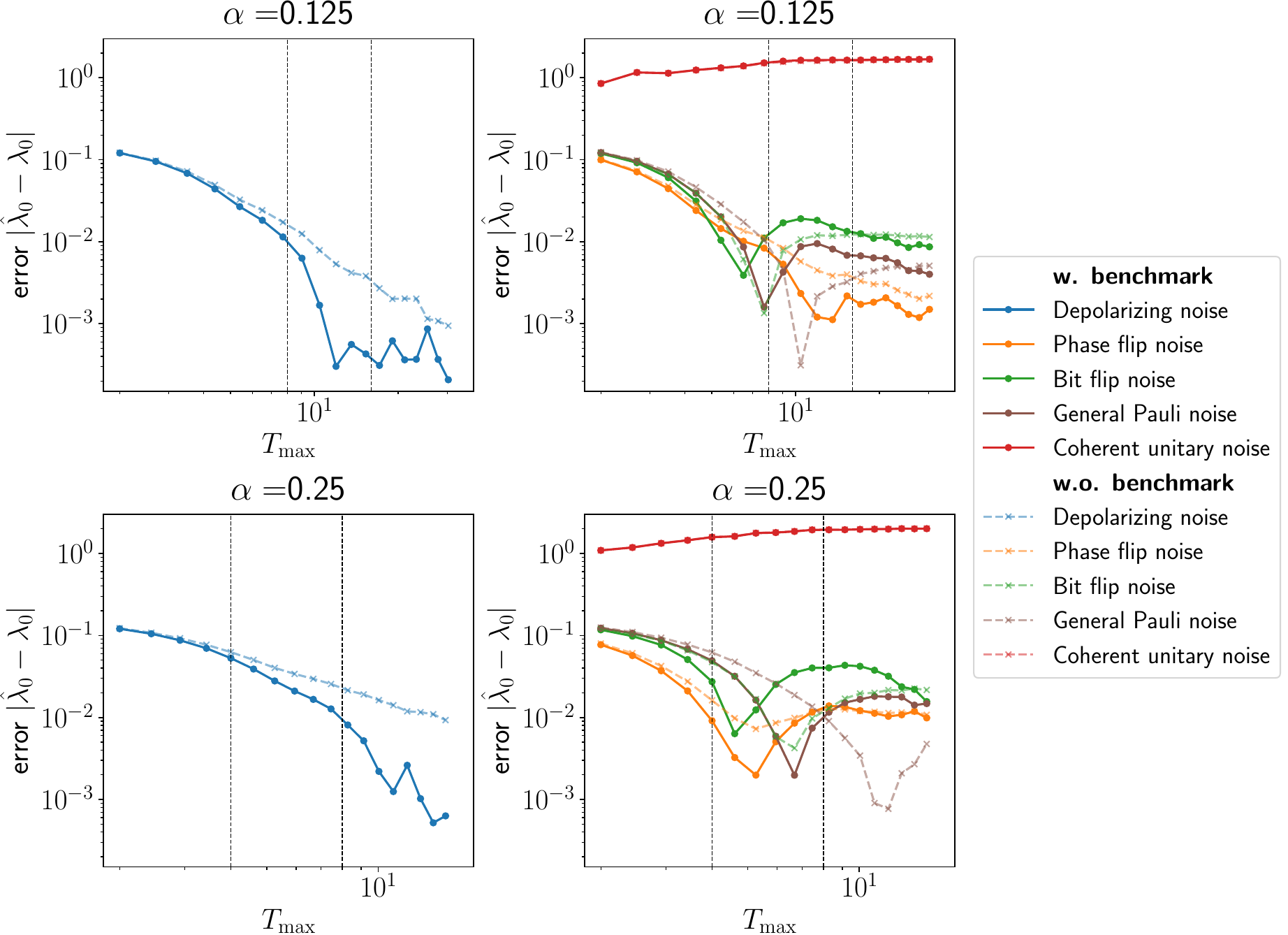}
    \caption{Estimation error of the ground-state energy. The vertical dashed line stands for $T_\mathrm{max} = 1 / \alpha, 2 / \alpha$. Each data point is an average of three independent repetitions.}
    \label{fig:tfim_regression}
\end{figure}

\section{Conclusion}

This paper introduces a ground-state energy estimation algorithm designed to withstand the challenges posed by the global depolarizing noise channel. The algorithm demonstrates robustness while maintaining a polynomial cost in $1/\epsilon$. To this end, it is essential for the quantum algorithm to estimate the ground-state energy within a maximum runtime of $\Or(\log \epsilon^{-1})$ in the noiseless setting. Failure to do so would result in the signal being overwhelmed by noise, necessitating an exponentially increasing number of samples for signal recovery. Existing algorithms like QPE and RPE do not meet this requirement. Through proper modifications to the QCELS algorithm, we have successfully achieved a maximal runtime that meets the aforementioned stringent requirement. 
Numerous numerical results demonstrate the efficiency of our algorithm in addressing not only global depolarizing noise but also other types of quantum noise such as local depolarizing noise.

Finally, we note that this work opens the door to several intriguing problems and unexplored directions:
\begin{itemize}
    \item Although Algorithm \ref{alg:main} can achieve $\epsilon$ accuracy with a total running time that scales polynomially in $\epsilon$, the exponent of this polynomial depends on the ratio $\alpha/\Delta$. To maintain this exponent as a constant, it is necessary for the depolarizing constant $\alpha$ to be on the order of $\mathcal{O}(\Delta)$. In scenarios where the spectral gap of $H$ decreases with the system size, a decreasing noise constant $\alpha$ is also required for accurate estimation of $\lambda_0$. An intriguing avenue for future investigation is whether substantial improvement is possible to reduce the dependence on $\alpha/\Delta$.

    \item  In Theorem \ref{thm_alg_new_informal}, the assumption $p_0 > 0.5$ is made to ensure that the solution of the optimization problem~\eqref{eqn:fixed_beta_new} closely approximates $\lambda_0$. If this condition is not met, similar to the approach outlined in~\cite{DingLin2023,WangStilkFrancaEtAl2022quantum}, a rough estimation of the ground state energy is required in the beginning. Specifically, we can employ the Fourier filtering technique introduced in \cite{LinTong2022} on the sequence $\left\{(t_n, \exp\left(\widetilde{\alpha} |t_n|\right)Z_n)\right\}$ to find a rough estimation $\widetilde{\lambda}_0$ such that $|\lambda_0-\widetilde{\lambda}_0|\leq \frac{\Delta}{2}$ with $T_{\max}=\Theta(1/\Delta)$. Subsequently, we can apply \cref{alg:main}, adjusting the optimization problem to be solved with $\theta\in[\widetilde{\lambda}_0-\Delta/2, \widetilde{\lambda}_0+\Delta/2]$. The resulting value $\theta^*$ obtained from this modified application is still an accurate approximation to $\lambda_0$. 
        
    \item Although the majority of prior ground state energy estimation algorithms necessitate $T_{\max}$ to scale as $\mathrm{poly}(1/\epsilon)$, there have been recent demonstrations of alternative algorithms with a maximal runtime scaling as $\wt{\Or}(\Delta^{-1})$ \cite{DingLin2023,WangStilkFrancaEtAl2022quantum}. While these algorithms can be more complex than our current scheme, we anticipate that certain adaptations can be made to enhance their resilience to the global depolarizing noise.
    \item It is important to note that real quantum computers encounter additional complex sources of noise, such as readout errors, and compiler optimization effects. In addition, the noise channel may even exhibit time heterogeneity. Quantum error correction may be ultimately necessary to address these challenges. While the ability to correct arbitrary errors is a theoretical triumph, early fault-tolerant quantum computers may exhibit structured errors. Exploiting the structure of noise becomes crucial in this context, and it may be important to combine error correction and algorithmic-level error mitigation to develop efficient quantum algorithms. For example, although our algorithm was initially developed
assuming global depolarizing noise, numerical simulations have demonstrated its effectiveness and
resilience in handling various other noise models, particularly local depolarizing noise. This observation motivates further research into improved theoretical understanding of ground-state energy estimation algorithms under local depolarizing noise, as well as efficient 
techniques for tailoring noise into local depolarizing noise for improving the efficiency of quantum algorithms on early fault-tolerant quantum computers.
\end{itemize}

\noindent{\large \textbf{Data availability}}\\
The experimental data that support the finding are available from the authors upon request.\\

\noindent{\large \textbf{Code availability}}\\
The codes that support the finding are available from the authors upon request.\\

\noindent {\large \textbf{Acknowledgments}}\\

This material is based upon work supported by the U.S. Department of Energy, Office of Science, National Quantum Information Science Research Centers, Quantum Systems Accelerator (Z.D.). Additional support is acknowledged from NSF Quantum Leap Challenge Institute (QLCI) program under Grant number OMA-2016245 (Y.D.), 
the Applied Mathematics Program of the US Department of Energy (DOE) Office of Advanced Scientific Computing Research under contract number DE-AC02-05CH1123, and the Google Quantum Research Award.  Y.T. acknowledges funding from U.S. Department of Energy Office of Science, Office of Advanced Scientific Computing Research (DE-SC0020290), and from U.S. Department of Energy, Office of Science, Basic Energy Sciences, under Award Number DE-SC0019374. Work supported by DE-SC0020290 is supported by the DOE QuantISED program through the theory consortium ``Intersections of QIS and Theoretical Particle Physics'' at Fermilab. The Institute for Quantum Information and Matter is an NSF Physics Frontiers Center. L.L. is a Simons Investigator in Mathematics.\\

\noindent {\large \textbf{Author contributions}}\\

 L.L. and Y.T. conceived the original study. Z.D. and L.L. carried out the theoretical analysis to support the study. Y.D. and Z.D. carried out the numerical simulation to support the study. All authors discussed the results of the manuscript, and contributed to the writing of the manuscript.\\

\noindent{\large \textbf{Competing interests}}\\
\noindent The authors declare no competing interests.

\bibliographystyle{abbrv}
\bibliography{ref}
\newpage

\appendix

\setcounter{equation}{0}
\setcounter{figure}{0}
\setcounter{table}{0}
\renewcommand{\figurename}{Supplementary Figure}
\renewcommand{\tablename}{Supplementary Table}
\renewcommand{\thefigure}{S\arabic{figure}}
\renewcommand{\thetable}{S\arabic{table}}

\begin{center}
    {\bf \Large Supplementary Information}
\end{center}

\renewcommand{\thesubsection}{\thesection.\arabic{subsection}}

\section{Detailed algorithm}

Throughout the paper, we use the following asymptotic notations besides the usual $\Or$ notation: we write $f=\Omega(g)$ if $g=\Or(f)$; $f=\Theta(g)$ if $f=\Or(g)$ and $g=\Or(f)$; $f=\wt{\Or}(g)$ if $f=\Or(g\operatorname{polylog}(g))$.
This section outlines the process of constructing $\widetilde{\alpha}$ and provides a detailed introduction to the data generator and the optimization problem.

\subsection{Estimation of \texorpdfstring{$\alpha$}{Lg}}\label{sec:estimation_beta}
To construct a proper loss function~\eqref{eqn:fixed_beta_new}, obtaining a $\widetilde{\alpha}$ well approximating the noise parameter $\alpha$ is an essential step. This problem has been tackled in previous work, such as in \cite{PhysRevLett.127.270502}, where the authors proposed a noise-estimation circuit to estimate $\alpha$ and use it to correct the output of the target circuit. 

Our approach corrects the noisy data using $\exp\left(\widetilde{\alpha}|t|\right)$ estimated from an auxiliary circuit whose structure is similar to the working Hadamard test circuits but whose justification does not need hard classical computation. Based on the observation in \eqref{eqn:Bt}, for each $t_{B,n}$, we couple $\mathcal{B}_{t_{B,n}}$ into the Hadamard test and average the outcome of the measurements of the ancilla qubits to obtain an approximation to $\exp(-\alpha|t_{B,n}|)$. Specifically, we repeat the quantum circuit in \cref{fig:flowchart} $N_{s,1}$ times with $t=t_B$ and define
\[
b_{n,m}=\left\{
\begin{aligned}
&1,\quad \text{outcome}=0\\
&-1,\quad \text{outcome}=1
\end{aligned}\right.\,,\quad B_{n}=\frac{1}{N_{s,1}}\sum^{N_{s,1}}_{m=1}b_{n,m}\,.
\]
The quantum circuit in \cref{fig:flowchart} is referred to as quantum benchmarking circuit. The detailed algorithm is summarized in \cref{alg:estimation_alpha}.
\begin{breakablealgorithm}
      \caption{Estimation of $\alpha$}
  \label{alg:estimation_alpha}
  \begin{algorithmic}[1]
  \State \textbf{Preparation:} Times : $\{t_{B,n}\}^{N_B}_{n=1}$, Number of data samples: $N_{s,1}$; 
  \State Run the benchmarking circuits in \cref{fig:flowchart} $N_{s,1}$ times with $\{t_{B,n}\}^{N_B}_{n=1}$ to obtain $\{B_{n}\}^{N_B}_{n=1}$.
    \State $
\left(\widetilde{\alpha},\widetilde{\beta}\right)\gets\mathrm{argmin}_{(a,b)}\sum^{N_B}_{n=1}\left|\log\left(B_{n}\right)+a|t_{B,n}|+b\right|^2$\Comment{Linear regression}
\State  \textbf{Output:} $\widetilde{\alpha}$
    \end{algorithmic}
\end{breakablealgorithm}

Finally, noticing that the random variable $B_{n}$ only takes the value on $1$ and $-1$, we have the following lemma that proves the first part of Theorem \ref{thm_alg_new_informal}:
\begin{lem}\label{lem:beta} Given any $0<\epsilon,\eta<1$, we can choose $N_B=\mathcal{O}(1)$, $|t_{B,n}|=\mathcal{O}(\alpha^{-1})$. Define \[\kappa=\left\|\begin{pmatrix}
N_B & \sum^{N_B}_{n=1}|t_{B,n}|\\
 \sum^{N_B}_{n=1}|t_{B,n}| &  \sum^{N_B}_{n=1}t^2_{B,n}
\end{pmatrix}^{-1}\right\|_2\,.\] Then, we can further set
$N_{s,1}=\Omega(\kappa\log(N_B/\eta)N^2_B/\epsilon^2)$ to obtain $|\widetilde{\alpha}-\alpha|\leq \epsilon$ with probability $1-\eta$.
\end{lem}
\begin{proof} We first notice
\[
\mathbb{P}(b_{n,m}=1)=\frac{1+\exp(-\alpha |t_{B,n}|)}{2},\quad \mathbb{P}(b_{n,m}=-1)=\frac{1-\exp(-\alpha |t_{B,n}|)}{2}\,,
\]
which implies
\[
\mathbb{E}(b_{n,m})=\exp(-\alpha |t_{B,n}|)\,.
\]
Since $|b_{n,m}|$ is bounded by $1$, using Hoeffding's inequality, we obtain
\[
\mathbb{P}\left(\left|\frac{1}{N_{s,1}}\sum^{N_{s,1}}_{m=1}b_{n,m}-\exp(-\alpha |t_{B,n}|)\right|\geq \frac{\exp(-\alpha |t_{B,n}|)\epsilon}{10 N_B}\right)\leq \eta/N_B\,,
\]
which implies
\[
\mathbb{P}\left(\sup_{n}\left|\log(B_n)+\alpha |t_{B,n}|\right|\geq \frac{\epsilon}{N_B}\right)\leq \eta\,.
\]
Define $R_n=\log(B_n)+\alpha |t_{B,n}|$. According to the formula of linear regression and the conditions of the lemma, we have
\[
\left\|\begin{pmatrix}
\widetilde{\alpha}-\alpha \\
\widetilde{\beta}
\end{pmatrix}\right\|=\left\|\begin{pmatrix}
N_B & \sum^{N_B}_{n=1}|t_{B,n}|\\
 \sum^{N_B}_{n=1}|t_{B,n}| &  \sum^{N_B}_{n=1}t^2_{B,n}
\end{pmatrix}^{-1}\begin{pmatrix}
\sum^{N_B}_{n=1}R_n \\
\sum^{N_B}_{n=1}|t_{B,n}|R_n
\end{pmatrix}\right\|\leq \epsilon\,.
\]
This concludes the proof.
\end{proof}

\subsection{Data Generation}\label{sec:mqc}
Next, we focus on the generation of the dataset. Using the quantum circuit in flowchart \cref{fig:flowchart}, we may 
\begin{itemize}
    \item Set $W=I$, measure the ancilla qubit and define a random variable $X_n$ such that $X_n=1$ if the outcome is $0$ and $X_n=-1$ if the outcome is $1$. Then 
    \begin{equation}\label{eqn:X}
    \mathbb{E}(X_n)=\exp(-\alpha|t|)\mathrm{Re}\left(\braket{\psi|\exp(-itH)|\psi}\right)\,.
    \end{equation}
    \item Set $W=S^\dagger$, measure the ancilla qubit and define a random variable $X_n$ such that $X_n=1$ if the outcome is $0$ and $X_n=-1$ if the outcome is $1$. Then 
    \begin{equation}\label{eqn:Y}
    \mathbb{E}(Y_n)=\exp(-\alpha|t|)\mathrm{Im}\left(\braket{\psi|\exp(-itH)|\psi}\right)\,.
    \end{equation}
\end{itemize}


Given a set of time points $\{t_n\}^{N_t}_{n=1}$ drawn from the probability density $a(t)$, we prepare the following data set:
\begin{equation}\label{eqn:dataset}
    \mathcal{D}_{H}=\left\{\left(t_n,Z_n\right)\right\}^{N_t}_{n=1}\,.
\end{equation}
where $Z_n=\frac{1}{N_{s,2}}\sum^{N_{s,2}}_{i=1} \left(X_{n,i}+iY_{n,i}\right)$. Here $X_{n,i},Y_{n,i}$ are independently generated by the quantum circuit (flowchart \cref{fig:flowchart}) with different $W$ and satisfy \cref{eqn:X,eqn:Y} respectively. Hence, we have
\begin{equation}\label{eqn:Zn_expect}
\mathbb{E}(Z_n)=\exp(-\alpha|t_n|)\left\langle\psi\right|\exp(-i t_n H)\ket{\psi},\quad |Z_n|\leq 2\,.
\end{equation}
We also note that if we use the above method to prepare the data set in \cref{eqn:dataset}, the maximal simulation time is $T_{\max}=\max_{1\leq n\leq N}|t_n|$ and the total simulation time is $\sum^N_{n=1}|t_n|$. The detailed algorithm for the data generator is summarized in \cref{alg:data}.

\begin{breakablealgorithm}
      \caption{Data generator}
  \label{alg:data}
  \begin{algorithmic}[1]
  \State \textbf{Preparation:} Number of time samples: $N_t$; Number of data samples: $N_{s,2}$; probability density: $a(t)$; 
 \State $n\gets 1$;
  \While{$n\leq N_t$}
  \State Generate a random variable $t_n$ with the probability density $a(t)$.
  \State Run the quantum circuit (flowchart \cref{fig:flowchart}) $N_{s,2}$ times with $t=t_n$ and $W=I$ to obtain $X_{n,i}$.
  \State Run the quantum circuit (flowchart \cref{fig:flowchart}) $N_{s,2}$ times with $t=t_n$ and $W=S^\dagger$ to obtain $Y_{n,i}$.
  \State $Z_{n}\gets \frac{1}{N_{s,2}}\sum^{N_{s,2}}_{i=1}\left(X_{n,i}+iY_{n,i}\right)$.
  \State $n\gets n+1$
  \EndWhile
    \State \textbf{Output:} $\left\{(t_n,Z_{n})\right\}^{N_t}_{n=1}$
    \end{algorithmic}
\end{breakablealgorithm}

\subsection{Loss function and main algorithm}\label{sec:main_algorithm_appendix}
Once the data set and an approximation $\widetilde{\alpha}$ have been generated, we define the numerical loss function as follows:
\[
L\left(r,\theta\right)=\frac{1}{N_t}\sum^{N_t}_{n=1}\left|\exp(\widetilde{\alpha} |t_n|)Z_n-r\exp(-i\theta t_n)\right|^2\,.
\]
We then proceed to minimize this loss function $L\left(r,\theta\right)$ to obtain an estimation of $\lambda_0$. The resulting solution of the minimization problem \eqref{eqn:op} can be a good approximation of the solution of \eqref{eqn:op_perfect}.

Define the expectation error $E_n=Z_n-\exp(-\alpha|t|)\sum^{M-1}_{m=0}p_m\exp(-i \lambda_m t_n)$. Note that $|E_n|$ is bounded by $3$ but may not be small if $N_{s,1}$ is not large. On the other hand, the expectation of $E_n$ is zero. Thus, when $N_t\gg 1$, we can use the central limit theorem to show that $\left|\frac{1}{N_t}\sum^{N_t}_{n=1}E_n\right|\leq \frac{1}{\sqrt{N_t}}$. Using this fact, we obtain
\begin{equation}\label{eqn:equivalent_op}
\begin{aligned}
&\mathrm{argmin}_{r,\theta}L\left(r,\theta\right)\\
=&\mathrm{argmin}_{r,\theta}\frac{1}{N_t}\sum^{N_t}_{n=1}\left|\exp(\widetilde{\alpha} |t_n|)Z_n-r\exp(-i\theta t_n)\right|^2\\
=&\mathrm{argmin}_{r,\theta}\frac{1}{N_t}\sum^{N_t}_{n=1}\left|\exp((\widetilde{\alpha}-\alpha)|t_n|)\sum^{M-1}_{m=0}p_m\exp(-i \lambda_m t_n)-r\exp(-i\theta t_n)\right|^2\\
&-\frac{2}{N_t}\sum^{N_t}_{n=1}\mathrm{Re}\left(\left\langle \exp(\widetilde{\alpha} |t_n|)E_n,r\exp(-i\theta t_n)\right\rangle\right)\\
\approx &\mathrm{argmin}_{r,\theta}\int^\infty_{-\infty}a(t)\left|\exp((\widetilde{\alpha}-\alpha)|t|)\sum^{M-1}_{m=0}p_m\exp(-i \lambda_m t)-r\exp(-i\theta t)\right|^2\rd t\\
=&\mathrm{argmin}_{r,\theta}\mathcal{L}\left(r,\theta\right)\,.
\end{aligned}
\end{equation}
Here, $\left\langle a,b\right\rangle=a^\dagger b$ for $a,b\in\mathbb{C}$. In the second equality, we have omitted terms that are independent of $(r,\theta)$. In the approximation step, we use $\frac{2}{N_t}\sum^{N_t}_{n=1}\mathrm{Re}\left\langle E_n,r\exp(-i\theta t_n)\right\rangle\approx 0$ when $N_t\gg 1$. The rigorous error bound for this approximation is shown in \cref{sec:proof_of_thm} \cref{lem:bound_error_new_loss}. From \cref{eqn:equivalent_op}, we observe that the optimization problem in \cref{eqn:op} can yield a solution that approaches the solution in \cref{eqn:op_perfect} with the ideal loss when $N_t\gg 1$. This derivation shows that solving \cref{eqn:op} for sufficiently large values of $N_t$ yields an approximate solution to the ideal optimization problem given by \cref{eqn:op_perfect}, which is expected to provide a good approximation to the ground state energy $\lambda_0$.  

\section{Complexity of Algorithm \ref{alg:main}}\label{sec:complexity_algorithm}

Based on \cref{alg:estimation_alpha} and \cref{lem:beta}, we have already established an efficient method for obtaining an accurate estimation of $\alpha$. In this section, we will investigate the complexity analysis of Algorithm \ref{alg:main}, assuming that $\tilde{\alpha}\approx \alpha$.

\subsection{Informal complexity analysis}\label{sec:ica}
Before introducing the rigorous complexity result of Algorithm \ref{alg:main}, we start with an informal derivation to demonstrate that by solving the optimization problem \eqref{eqn:op}, it is possible to obtain an accurate approximation of the dominant eigenvalues with a short maximal running time and polynomial total running time.

Define the spectral gap $\Delta=\lambda_1-\lambda_0$. In particular, we are interested in the case when $\Delta$ and $\alpha$ are much larger than the precision $\epsilon$. Assume 
\begin{equation}\label{eqn:pre_condition_intuitive}
q=p_0-0.5>0,\ T=\widetilde{\Omega}\left(\frac{1}{\Delta}\right), \gamma=\Omega\left(\log^{\frac{1}{2}}(q)\right),\ N_t=\Omega\left(\frac{1}{q^2}\exp(2\gamma \alpha T)\right),\  N_{s,2}=\Theta(1)\,.
\end{equation} 
Here $\gamma$ is the truncation parameter in the definition of the truncated Gaussian.

Let $F^*(x)=\exp\left(-\frac{T^2x^2}{2}\right)$. Assuming $\widetilde{\alpha}\approx \alpha$, the loss function \cref{eqn:fixed_beta_new} can be simplified as
\[
L\left(r,\theta\right)=\frac{1}{N_t}\sum^{N_t}_{n=1}\left|\exp\left(\alpha |t_n|\right)Z_n-r\exp(-i\theta t_n)\right|^2
,.
\]
For a fixed value of $\theta$, we can see that the loss function $L(r, \theta)$ is a quadratic function in $r$. This implies that the optimal value of $r$ can be found and the optimization problem can be rewritten as
\begin{equation}\label{eqn:new_op_intuitive}
\begin{aligned}
\theta^*&=\mathrm{argmax}_{\theta\in[-\pi,\pi]}\left|\frac{1}{N_t}\sum^{N_t}_{n=1}\exp(\alpha |t_n|)Z_n\exp(i\theta t_n)\right|^2\\
&=\mathrm{argmax}_{\theta\in[-\pi,\pi]}\left|\sum^{M-1}_{m=0}p_m\int^\infty_{-\infty}\frac{1}{\sqrt{2\pi}T}\exp\left(-\frac{t^2}{2T^2}\right)\exp(-i(\lambda_m-\theta)t)\rd t+E(\theta)\right|^2\\
&=\mathrm{argmax}_{\theta\in[-\pi,\pi]}\left|\sum^{M-1}_{m=0}p_m F^*(\theta-\lambda_m)+E(\theta)\right|
\end{aligned}\,,
\end{equation}
where $E(\theta)$ denotes the error term that comes from finite sampling of $t$ and shot noise. 
Lemma \ref{lem:bound_error_new_loss} provides a detailed analysis of this error term, which satisfies
\begin{equation}\label{eqn:E_intuitive_bound}
|E(\theta)|=\mathcal{O}\left(\exp(\gamma \alpha T)/\sqrt{N_t}+\exp(-\gamma^2)\right)\,.
\end{equation}

Let $f(\theta)=\sum^{M-1}_{m=0}p_m F^*(\theta-\lambda_m)$. First, note that $f(\lambda_0)\geq p_0$ since $F^*(\lambda_0-\lambda_0)=1$. Using the formula for $F^*$, we obtain
\[
f(\theta^*)\leq p_0\exp(-T^2|\theta^*-\lambda_0|^2/2)+(1-p_0)\,.
\]
Because $\theta^*$ is the maximal point of $|f(\theta)+E(\theta)|$, we should have 
\[
p_0-O(q)=f(\lambda_0)-|E(\lambda_0)|\leq f(\theta^*)+|E(\theta^*)|\leq p_0\exp(-T^2|\theta^*-\lambda_0|^2/2)+(1-p_0)+O(q)\,,
\]
where we use \cref{eqn:pre_condition_intuitive} to obtain $|E(\lambda_0)|,|E(\theta^*)|=O(q)$. This implies 
\[
T|\theta^*-\lambda_0|=\mathcal{O}\left(\log^{1/2}\left(\frac{p_0}{q}\right)\right)\,.
\]
Since $T=\Omega\left(1/\Delta\right)$, we have $|\theta^*-\lambda_0|<O\left(1/T\right)<\Delta/2$. This implies that $\theta^*$ falls into the interval $[\lambda_0-\Delta/2, \lambda_0+\Delta/2]$. Therefore, for any $m\geq 1$, $|\theta^*-\lambda_m|>\Delta/2$. These inequalities help us refine the upper bound of $f(\theta^*$):
\[
\begin{aligned}
    f(\theta^*)&\leq p_0\exp(-T^2|\theta^*-\lambda_0|^2/2)+(1-p_0)\exp(-T^2\Delta^2/8)\\
\end{aligned}
\]
Because $\theta^*$ is the maximal point of $|f(\theta)+E(\theta)|$,  similar to before we obtain
\[
\frac{p_0-(1-p_0)\exp(-T^2\Delta^2/8)-|E(\lambda_0)|-|E(\theta^*)|}{p_0}\leq \exp(-T^2|\theta^*-\lambda_0|^2/2)\,,
\]
which gives us  
\[
\left|\theta^*-\lambda_0\right|= \Or\left(\frac{\sqrt{(1-p_0)\exp(-T^2\Delta^2/8)+|E(\lambda_0)|+|E(\theta^*)|}}{T}\right)\,.
\]
By choosing $T\ge \frac{4}{\Delta}\log^{1/2}\left(\frac{1}{\epsilon}\right)$, we have $\exp(-T^2\Delta^2/8)\le \epsilon^2$. Then when $|E(\lambda_0)|+|E(\theta^*)|=\widetilde{\Or}(\epsilon^2)$, we obtain $\left|\theta^*-\lambda_0\right|\leq \epsilon$. Finally, according to \cref{eqn:E_intuitive_bound}, we choose 
\[
\gamma=\Theta\left(\log^{\frac{1}{2}}\left(\frac{1}{\epsilon}\right)\right),\quad N_t=\Theta\left(\left(\frac{1}{\epsilon}\right)^{\Theta(1+\frac{\alpha}{\Delta})}\right)\,. 
\]
to bound the error terms $|E(\theta)|$. Therefore, to achieve $\epsilon$-accuracy, we set
\[
T_{\max}=\gamma T=\Theta\left(\frac{1}{\Delta}\log\left(\frac{1}{\epsilon}\right)\right),\quad T_{\mathrm{total}}=\mathrm{poly}\left(\frac{1}{\epsilon}\right)\,.
\]

\subsection{Rigorous complexity of Algorithm \texorpdfstring{\ref{alg:main}}{Lg}}
Now, we are ready to introduce the rigorous complexity result of Algorithm \ref{alg:main}, which is summarized in the following theorem. 
\begin{thm}\label{thm_alg_new}
Assume $p_0>0.5$ and define $q=p_0-0.5$. Given any $0<\epsilon<1/2$ and $0<\eta,o<1$, we assume 
\begin{equation}\label{eqn:beta_close_to_alpha}
|\widetilde{\alpha}-\alpha|=\mathcal{O}\left(\min\left\{\epsilon^{1+o(1)},\Delta q\right\}\log^{-\frac{1}{2}}\left(\frac{1}{\epsilon}\right)\right)\,.
\end{equation}
In Algorithm \ref{alg:main}, we can choose 
\[
T=\Theta\left(\frac{1}{\Delta}\log^{\frac{1}{2}}\left(\frac{1}{\epsilon}\right)\right),\ \gamma= \Theta\left(\log^{\frac{1}{2}}\left(\frac{1}{\min\{\epsilon/\Delta,q\}}\right)\right)\,,
\]
and
\[
N_t=\Theta\left(\max\left\{\frac{\Delta^2}{\epsilon^{2+o}},\frac{1}{ q^2}\right\}\left(\frac{\alpha}{\Delta}+1\right)\left(\frac{1}{\epsilon}\right)^{\Theta\left(\frac{\alpha}{\Delta}\right)}\log\left(\frac{\log(1/o)}{\eta\epsilon q\Delta}\right)\right),\quad N_{s,2}=\Theta(1)\,.
\]
to guarantee that with probability $1-\eta$, 
\[
|\theta^*-\lambda_0|\leq \epsilon\,.
\]
In particular, when $\epsilon$ is small enough, we have
\[
T_{\max}=\Theta\left(\frac{1}{\Delta}\log\left(\frac{1}{\epsilon}\right)\right)\,,
\]
and 
\[
T_{\mathrm{total}}=\Theta\left(\left(\alpha+\Delta\right)\left(\frac{1}{\epsilon}\right)^{2+o+\Theta\left(\frac{\alpha}{\Delta}\right)}\log\left(\frac{\log(1/o)}{\eta\epsilon\Delta}\right)\right)\,.
\]
\end{thm}

We emphasize that Theorem \ref{thm_alg_new_informal} is a direct result of Lemma \ref{lem:beta} and Theorem \ref{thm_alg_new}. We put the proof of \cref{thm_alg_new} in \cref{sec:proof_of_thm}. The proof follows a similar approach to the previous intuitive derivation. To show that $|\theta^*-\lambda_0|\leq \epsilon$, we need to take two steps. Firstly, as mentioned earlier, we must control $E(\theta)$ by increasing the number of samples $N_t,N_{s,2}$ and selecting an appropriate $\gamma$ such that it does not significantly alter the loss function. Secondly, similar to the previous intuitive derivation, we demonstrate that $|\theta^*-\lambda_0|<\Delta/2$ and $|\theta^*-\lambda_m|>\Delta/2$ for $m\geq1$, indicating that $\theta^*$ and $\lambda_m$ are well-separated in the loss function. Once this separation property is established, we leverage the assumption that $T=\Omega\left(1/\Delta\right)$ to show that the pollution terms from other eigenvalues (i.e., $F(\theta^*-\lambda_m)$ in \cref{eqn:new_op_intuitive}) are $\widetilde{\mathcal{O}}(\epsilon^2)$. This step enables us to prove that $\theta^*$ is $\epsilon$-close to $\lambda_0=\argmax_{\theta} F(\theta-\lambda_0)$.

Finally, it is worth noting that, to guarantee that the total runtime of the algorithm scales as $\mathrm{poly}(1/\epsilon)$ with depolarizing noise, it is crucial to keep $T_{\max}$ relatively small due to the exponentially increasing variance of the recovering data point with respect to $t$.  Specifically, we can demonstrate that $Z_t$ is an unbiased estimation of $\exp(-\alpha|t|)\left\langle\psi\right|\exp(-it H)\ket{\psi}$ and the variance of this estimation has a lower bound independent of $\alpha$. Therefore, when we use $\exp(\widetilde{\alpha}|t|)Z_t$ to approximate $\exp((\widetilde{\alpha}-\alpha)|t|)\left\langle\psi\right|\exp(-it H)\ket{\psi}$, the variance of this estimation is $\Omega(\exp(2\widetilde{\alpha} |t|))$.
If $\widetilde{\alpha}\approx \alpha$, then the number of samples required to ensure a small estimation error will be $\Omega(\exp(2\alpha t))$. This means that $t$ must not increase faster than $\log(1/\epsilon)$ to guarantee that the total runtime scales as $\mathrm{poly}(1/\epsilon)$. 

\subsection{Proof of Theorem \ref{thm_alg_new}}\label{sec:proof_of_thm}


To make the derivation in the previous section rigorous, we consider the error in $\alpha$, as well as proving that $E(\theta)$ is uniformly small in $\theta$.

In this section, we assume that $\left|\widetilde{\alpha}-\alpha\right|\leq \delta<\alpha$ and define $F^*(x)=\exp\left(-\frac{T^2x^2}{2}\right)$. First, similar to \eqref{eqn:new_op_intuitive}, minimizing $L\left(r,\theta\right)$ to obtain $\theta^*$ is equivalent to
\begin{equation}\label{eqn:theta_star_new_beta}
\begin{aligned}
\theta^*&=\mathrm{argmax}_{\theta\in[-\pi,\pi]}\left|\frac{1}{N_t}\sum^{N_t}_{n=1}\exp\left(\widetilde{\alpha} |t_n|\right)Z_n\exp(i\theta t_n)\right|^2\\
&=\mathrm{argmax}_{\theta\in[-\pi,\pi]}\left|\sum^{M-1}_{m=0}p_m\int^\infty_{-\infty}a(t)\exp(-i(\lambda_m-\theta)t)\rd t+E_1(\theta)+E_2(\theta)+E_3(\theta)\right|^2\\
&=\mathrm{argmax}_{\theta\in[-\pi,\pi]}\left|\sum^{M-1}_{m=0}p_m F(\theta-\lambda_m)+E_1(\theta)+E_2(\theta)+E_3(\theta)\right|
\end{aligned}
\end{equation}
where 
\[
\begin{aligned}
E_1(\theta)=&\sum^{M-1}_{m=0}p_m\left(\frac{1}{N_t}\sum^{N_t}_{n=1}\exp((\widetilde{\alpha}-\alpha)|t_n|)\exp(-i(\lambda_m-\theta)t_n)\right.\\
&\left.-\int^\infty_{-\infty}a(t)\exp((\widetilde{\alpha}-\alpha)|t|)\exp(-i(\lambda_m-\theta)t)\rd t\right)\\
&+\frac{1}{N_t}\sum^{N_t}_{n=1}E_n\exp(i\theta t_n)\exp(\widetilde{\alpha}|t_n|)\,,
\end{aligned}
\]
\[
E_2(\theta)=\sum^{M-1}_{m=0}p_m\int^\infty_{-\infty}a(t)(\exp((\widetilde{\alpha}-\alpha)|t|)-1)\exp(-i(\lambda_m-\theta)t)\rd t\,,
\]
and
\[
E_3(\theta)=\sum^{M-1}_{m=0}p_m\int^\infty_{-\infty}\left(a(t)-\frac{1}{\sqrt{2\pi}T}\exp\left(-\frac{t^2}{2T^2}\right)\right)\exp(-i(\lambda_m-\theta)t)\rd t\,.
\]
Here, the expectation error 
\begin{equation}\label{eqn:E_n}
E_n=Z_n-\exp(-\alpha|t_n|)\sum^{M-1}_{m=0}p_m\exp(-i \lambda_m t_n)
\end{equation}
satisfying $\mathbb{E}(E_n)=0$ and $|E_n|\leq 3$. We first give a lemma to bound the error terms $E_{1\leq i\leq 3}(\theta)$ when $N$ is large:
\begin{lem}[Bound of the error]\label{lem:bound_error_new_loss}
Given $0<\epsilon<1$, $0<\rho\leq \pi$, $T>1$, $\eta>0$, and $0<\delta<\min\{\alpha,1\}$, we assume $|\widetilde{\alpha}-\alpha|\leq \delta$. If $N_{s,2}=\Theta(1)$, 
\begin{equation}\label{eqn:condition_N_gamma}
N_t=\Omega\left(\frac{1}{\epsilon^2}\left[\log\left(\frac{\max\{\rho\gamma T,1\}}{\eta\epsilon}\right)+\alpha\gamma T\right]\exp(2\alpha\gamma T)\right),\ \gamma=\Omega\left(\log^{\frac{1}{2}}(1/\epsilon)\right),\ \delta=\mathcal{O}\left(\frac{\epsilon}{T}\right)\,,
\end{equation}
then
\begin{equation}\label{eqn:bound_E1}
\mathbb{P}\left(\sup_{|\theta-\lambda_0|\leq \rho}|E_1(\theta)|\geq \epsilon\right)\leq \eta,\quad \mathbb{P}\left(\sup_{|\theta-\lambda_0|\leq \rho}|E_1(\theta)-E_1(\lambda_0)|\geq \rho\gamma T\epsilon\right)\leq \eta\,,
\end{equation}
\begin{equation}\label{eqn:bound_E2}
\sup_{|\theta-\lambda_0|\leq \rho}\left|E_2(\theta)\right|\leq \epsilon,\quad \sup_{|\theta-\lambda_0|\leq \rho}\left|E_2(\theta)-E_2(\lambda_0)\right|\leq \rho\gamma T\epsilon\,,
\end{equation}
and
\begin{equation}\label{eqn:bound_E3}
\sup_{|\theta-\lambda_0|\leq \rho}\left|E_3(\theta)\right|\leq \epsilon^2\,.
\end{equation}
\end{lem}
\begin{proof}[Proof of Lemma \ref{lem:bound_error_new_loss}] In the proof, we always use $C$ to represent a uniform constant that is independent of other parameters and the constant $C$ might vary in different lines. Without loss of generality, assume $\lambda_0=0$.
 
First, because $|t_n|\leq \gamma T$ and  $\widetilde{\alpha}<\alpha+\delta<2\alpha$, for any $\theta_1,\theta_2\in \left[-\rho,\rho\right]$, we have
\[
\left|E_1(\theta_1)-E_1(\theta_2)\right|\leq 2\exp(2\alpha \gamma T)\gamma T|\theta_1-\theta_2|
\]
Next, since $|E_n\exp(\widetilde{\alpha} t_n)|<3\exp(2\alpha \gamma T)$ and $\left|\exp((\widetilde{\alpha}-\alpha)|t_n|)\exp(-i(\lambda_m-\theta)t_n)\right|\leq \exp(\delta \gamma T)$, using Hoeffding's inequality, for fixed $\theta\in [-\rho,\rho]$, we have
\[
\mathbb{P}\left(\left|E_1(\theta)\right|\geq \xi\right)\leq C\exp(-CN_t\exp(-2\alpha\gamma T)\xi^2),\quad \forall \xi\in\mathbb{R}_+\,.
\]
Given any $\chi>0$, we can find a set of $\lfloor\frac{2\rho}{\chi}\rfloor$ points $\{\theta_i\}^{\lfloor\frac{2\rho}{\chi}\rfloor}_{i=1}$ such that for any $|\theta|\leq \rho$, there exists $i$ such that $|\theta_i-\theta|\leq \chi$. Because $E_1(\theta)$ is a $2\exp(2\alpha \gamma T)\gamma T$-Lipschitz function, we have
\[
\mathbb{P}\left(\sup_{|\theta|\leq \rho}|E_1(\theta)|\geq \xi+2\exp(2\alpha \gamma T)\gamma T \chi\right)\leq \mathbb{P}\left(\sup_{i}|E_1(\theta_i)|\geq \xi\right)\leq \frac{C\rho}{\chi}\exp(-CN_t\exp(-2\alpha\gamma T)\xi^2)
\]
for any $0<\chi<1$.

Choosing $\xi=\frac{\epsilon}{2},\chi=\frac{\epsilon}{4\exp(2\alpha \gamma T)\gamma T}$ and using \cref{eqn:condition_N_gamma}, we have
\[
\mathbb{P}\left(\sup_{|\theta|\leq \rho}|E_1(\theta)|\geq \epsilon\right)\leq \eta\,.
\]
This proves the first inequality of \cref{eqn:bound_E1}. To prove the second inequality, define
\[
D(\theta)=E_1(\theta)-E_1(\lambda_0)\,.
\]
Because $|t_n|\leq \gamma T$ and $|\theta|\leq \rho$, we notice that
\[
\left|\exp(-i(\lambda_m-\theta)t_n)-\exp(-i(\lambda_m-\lambda_0)t_n)\right|\leq \rho\gamma T,\ \left|\exp(i\theta t_n)-\exp(i\lambda_0t_n)\right|\leq \rho\gamma T\,.
\]
Similar to previous proof, for fixed $\theta\in [-\rho,\rho]$, we have
\[
\mathbb{P}\left(\left|D(\theta)\right|\geq \xi\right)\leq C\exp\left(-\frac{CN_t\xi^2}{(\rho\gamma T)^2\exp(2\alpha\gamma T)}\right),\quad \forall \xi\in\mathbb{R}_+\,.
\]
Combining this with the fact that $D(\theta)$ is also a 
$2\exp(2\alpha \gamma T)\gamma T$-Lipschitz function, we obtain the second inequality of \cref{eqn:bound_E1}.

Next, to prove \cref{eqn:bound_E2}, since $|
\widetilde{\alpha}-\alpha|\leq \delta$, we first notice that for any $\theta$, 
\[
|E_2(\theta)|\leq \int^\infty_{-\infty} a(t)\left(\exp(\delta |t|)-1\right)dt=\frac{1}{\sqrt{2\pi}}\int^\infty_{-\infty}\exp\left(-\frac{t^2}{2}\right)\left(\exp(\delta T|t|)-1\right)dt
\]
Define $g(x)=\frac{1}{\sqrt{2\pi}}\int^\infty_{-\infty}\exp\left(-\frac{t^2}{2}\right)\left(\exp(x|t|)-1\right)dt$. It's straightforward to see that $g(x)$ is a smooth function in $x$. Becaseu $\delta T<\epsilon/C<1$ and $g(0)=0$, we have
\[
|E_2(\theta)|\leq g\left(\delta T\right)\leq \left(\sup_{x\in[0,1]}|g'(x)|\right)\tau\leq C\delta T<\epsilon\,,
\]
which proves the first inequality of \cref{eqn:bound_E2}. The second inequality is a direct result of the following inequality
\[
|E_2(\theta)-E_2(\lambda_0)|\leq \gamma T|\theta-\lambda_0|\int^\infty_{-\infty} a(t)\left(\exp(\delta |t|)-1\right)dt\leq \rho \gamma T\epsilon\,,
\]
where we use the fact that $\exp(-i(\lambda_m-\theta)t)$ is $\gamma T$-Lipschitz in $\theta$ when $|t|\leq \gamma T$.

Finally, notice
\[
\begin{aligned}
E_3(\theta)=&\sum^{M-1}_{m=0}\left(1-\frac{1}{A_\gamma}\right)p_m\int^{\gamma T}_{-\gamma T}a(t)\exp(-i(\lambda_m-\theta)t)\rd t\\
&+\sum^{M-1}_{m=0}p_m\int_{|t|>\gamma T}a(t)\exp(-i(\lambda_m-\theta)t)\rd t\,.
\end{aligned}
\]
Using the tail bounds of normal distribution, we obtain 
\[
A_\gamma=\int^{\gamma T}_{-\gamma T}a(t)dt=\int^\gamma_{-\gamma}\frac{1}{\sqrt{2\pi}}\exp\left(-\frac{x^2}{2}\right)dx\geq 1-\sqrt{\frac{2}{\pi}}\frac{\exp\left(-\gamma^2/2\right)}{\gamma}\geq 1-\frac{\epsilon^2}{4}\,,
\]
where we use $\gamma \geq C+2\log^{1/2}(1/\epsilon)$ in the last inequality. This also implies
\[
\left|\sum^{M-1}_{m=0}\left(1-\frac{1}{A_\gamma}\right)p_m\int^{\gamma T}_{-\gamma T}a(t)\exp(-i(\lambda_m-\theta)t)\rd t\right|\leq \left|1-\frac{1}{A_\gamma}\right|\leq \frac{\epsilon^2}{2}\,,
\]
and
\[
\left|\sum^{M-1}_{m=0}p_m\int_{|t|>\gamma T}a(t)\exp(-i(\lambda_m-\theta)t)\rd t\right|\leq \int_{|t|>\gamma T}a(t)dt\leq \frac{\epsilon^2}{2}\,.
\]
These two inequalities give us \cref{eqn:bound_E3}.
\end{proof}

Using Lemma \ref{lem:bound_error_new_loss}, we can first show the following proposition that is a weak version of Theorem \ref{thm_alg_new}:
\begin{prop}\label{prop_new_loss} Given $0<\epsilon_0\leq 1/2$ and $0<\eta<1$, we can choose
\begin{equation}\label{eqn:condition_T_prop}
T=\Omega\left(\frac{1}{\Delta}\log^{1/2}\left(\frac{1}{\epsilon_0}\right)\right),\ \gamma=\Omega\left( \log^{1/2}\left(\frac{1}{\min\{T\epsilon_0,q\}}\right)\right),\ \delta=\mathcal{O}\left(\min\left\{T\epsilon_0^2,\frac{q}{T}\right\}\right)\,,
\end{equation}
and
\begin{equation}\label{eqn:condition_N_prop}
\begin{aligned}
    &N_t=\Omega\left(\max\left\{\frac{1}{T^4\epsilon^4_0}\left[\log\left(\frac{\gamma T}{\eta\epsilon_0}\right)+\alpha\gamma T\right],\frac{1}{ q^2}\left[\log\left(\frac{\gamma T}{\eta q}\right)+\alpha\gamma T\right]\right\}\exp(2\alpha\gamma T)\right),\\
    &N_{s,2}=\Theta(1)
\end{aligned}
\end{equation}
to guarantee that with probability $1-\eta$, 
\begin{equation}\label{eqn:condition_theta_prop}
|\theta^*-\lambda_0|\leq \epsilon_0\,.
\end{equation}
\end{prop}
\begin{proof}[Proof of Proposition \ref{prop_new_loss}]  Define
\[
f(\theta)=\left|\sum^{M-1}_{m=0}p_m F(\theta-\lambda_m)+E_1(\theta)+E_2(\theta)+E_3(\theta)\right|\,.
\]
Using \cref{eqn:condition_T_prop} and \cref{eqn:condition_N_prop} in Lemma \ref{lem:bound_error_new_loss} (with $\epsilon=\min\{T^2\epsilon^2_0, q/4\}, \rho=\pi$), we can have
\[
\mathbb{P}\left(\sup_{|\theta|\leq \pi}|E_1(\theta)|+|E_2(\theta)|+|E_3(\theta)|\geq \min\left\{\frac{T^2\epsilon^2_0}{4},\frac{q}{4}\right\}\right)\leq \eta/4
\]

\textbf{Step 1: Show $|\theta^*-\lambda_0|\leq \Delta/2$.}

When $|E_1(\theta)|+|E_2(\theta)|+|E_3(\theta)|<\frac{q}{4}$, we first have 
\[
f(\lambda_0)\geq p_0-|E_1(\lambda_0|-|E_2(\lambda_0)|-|E_3(\lambda_0)|\geq p_0-\frac{q}{4}
\]
Assume $|\theta^*-\lambda_0|> \Delta/2$, since $p_0>0.5$, we have
\[
f(\theta^*)<0.5+p_0\exp\left(-\frac{T^2\Delta^2_{\lambda}}{8}\right)+|E_1(\theta^*)|+|E_2(\theta^*)|+|E_3(\theta^*)|< 0.5+\frac{q}{2}\,,
\]
where we use $T>\frac{C}{\Delta}\log^{1/2}\left(\frac{1}{\epsilon_0}\right)$ and $\epsilon_0<1/2$ in the last inequality. Thus, we have
\[
f(\theta^*)<0.5+\frac{q}{2}<p_0-\frac{q}{4}<f(\lambda_0)\,.
\]
This is contradicted to the fact that $\theta^*$ is the maximal point of $f(\theta)$. Thus, we have
\[
\mathbb{P}\left(|\theta^*-\lambda_0|\leq \Delta/2\right)\geq 1-\eta/4
\]
\textbf{Step 2: Show $|\theta^*-\lambda_0|\leq \epsilon_0$.}

When $|E_1(\theta)|+|E_2(\theta)|+|E_3(\theta)|<\frac{T^2\epsilon^2_0}{4}$ and $|\theta^*-\lambda_0|\leq \Delta/2$, we have
\[
\begin{aligned}
f(\theta^*)<&p_0\exp\left(-\frac{T^2(\theta^*-\lambda_0)^2}{2}\right)+\exp\left(-\frac{T^2\Delta^2_{\lambda}}{2}\right)+|E_1(\theta^*)|+|E_2(\theta^*)|+|E_3(\theta^*)|\\
\leq & p_0\exp\left(-\frac{T^2(\theta^*-\lambda_0)^2}{8}\right)+\frac{T^2\epsilon_0^2}{2}
\end{aligned}
\]
Because $f(\theta^*)>f(\lambda)$, we must have
\[
p_0\exp\left(-\frac{T^2(\theta^*-\lambda_0)^2}{2}\right)+\frac{T^2\epsilon^2_0}{2}>p_0-\frac{T^2\epsilon^2_0}{4}\,.
\]
Using $p_0>0.5$, we have
\[
\exp\left(-\frac{T^2(\theta^*-\lambda_0)^2}{2}\right)>1-\frac{3T^2\epsilon^2_0}{2}\,,
\]
which implies $|\theta^*-\lambda_0|\leq \epsilon_0$. This concludes the proof of \cref{eqn:condition_theta_prop}.
\end{proof}

We note that Proposition \ref{prop_new_loss} is a weaker version of \cref{thm_alg_new} since the power of $1/\epsilon_0$ in the condition of $N_t$ is $4$. Now, we start reducing this power to $2$. The idea is borrowed from \cite{DingLin2023}. We first have the following lemma:
\begin{lem}\label{lem_new_loss_improve} Given $0<\epsilon_0\leq 1/2$ and $0<\eta<1$, we assume $T,\gamma,\delta,N_t,N_{s,2}$ satisfy \cref{eqn:condition_T_prop} and \cref{eqn:condition_N_prop}. Given $\epsilon_1<\epsilon_0$, we can choose
\begin{equation}\label{eqn:condition_T_lem_2}
T=\Omega\left(\frac{1}{\Delta}\log^{1/2}\left(\frac{1}{\epsilon_1}\right)\right),\ \gamma=\Omega\left( \log^{1/2}\left(\frac{1}{T\epsilon_1}\right)\right),\ \delta=\mathcal{O}\left(\frac{\epsilon^2_1}{\gamma \epsilon_0}\right)\,,
\end{equation}
and
\begin{equation}\label{eqn:condition_N_lem_2}
N_t=\Omega\left(\frac{\gamma^2\epsilon^2_0}{T^2\epsilon^4_1}\left[\log\left(\frac{\gamma }{\eta\epsilon_1}\right)+\alpha\gamma T\right]\exp(2\alpha\gamma T)\right),\ N_{s,2}=\Theta(1)
\end{equation}
to guarantee that with probability $1-\eta$, 
\begin{equation}\label{eqn:condition_lem_2}
|\theta^*-\lambda_0|\leq \epsilon_1\,.
\end{equation}
\end{lem}
\begin{proof} According to Proposition \ref{prop_new_loss}, we first have 
\[
\mathbb{P}\left(|\theta^*-\lambda_0|\leq \epsilon_0\right)\geq 1-\frac{\eta}{8},\quad 
\mathbb{P}\left(\sup_{|\theta|\leq \pi}|E_1(\theta)|+|E_2(\theta)|+|E_3(\theta)|\geq \frac{T^2\epsilon^2_0}{4}\right)\leq \eta/8
\]
Plugging $\rho=\epsilon_0$ and $\epsilon=\frac{T\epsilon^2_1}{\gamma \epsilon_0}$ in Lemma \ref{lem:bound_error_new_loss} and using \cref{eqn:condition_T_lem_2}, we have
\[
\mathbb{P}\left(\sup_{|\theta-\lambda_0|\leq \epsilon_0}|E_1(\theta)-E_1(\lambda_0)|\geq CT^2\epsilon^2_1\right)\leq \eta/8,\ \sup_{|\theta-\lambda_0|\leq \epsilon_0}|E_2(\theta)-E_2(\lambda_0)|\leq  CT^2\epsilon^2_1\,,\] and \[\sup_{|\theta-\lambda_0|\leq \epsilon_0}|E_3(\theta)|\leq  CT^2\epsilon^2_1\,.
\]
When $|\theta^*-\lambda_0|\leq \epsilon_0$, using the above inequalities, we have
\begin{equation}\label{eqn:upperbound_f_theta_star}
\begin{aligned}
f(\theta^*)&=\left|p_0\exp\left(-\frac{T^2(\theta^*-\lambda_0)^2}{2}\right)+\sum^{M-1}_{m=1}p_m F(\theta^*-\lambda_m)+E_1(\theta^*)+E_2(\theta^*)+E_3(\theta^*)\right|\\
&\leq \left|p_0\exp\left(-\frac{T^2(\theta^*-\lambda_0)^2}{2}\right)+E_1(\theta^*)+E_2(\theta^*)\right|+CT^2\epsilon^2_1\\
&\leq \left|p_0\exp\left(-\frac{T^2(\theta^*-\lambda_0)^2}{2}\right)+E_1(\lambda_0)+E_2(\lambda_0)\right|+CT^2\epsilon^2_1
\end{aligned}
\end{equation}
where we use $T>\frac{C}{\Delta}\log^{1/2}\left(\frac{1}{\epsilon_1}\right)$ and $\left|E_3(\theta^*)\right|\leq CT^2\epsilon^2_1$ in the first inequality, $|E_1(\theta^*)-E_1(\lambda_0)|\leq CT^2\epsilon^2_1$ and $|E_2(\theta^*)-E_2(\lambda_0)|\leq CT^2\epsilon^2_1$ in the second inequality. In the meantime, we have
\begin{equation}\label{eqn:lowerbound_f_lambda_0}
\begin{aligned}
f(\lambda_0)&=\left|p_0+\sum^{M-1}_{m=1}p_m F(\lambda_0-\lambda_m)+E_1(\lambda_0)+E_2(\lambda_0)+E_3(\lambda_0)\right|\\
&\geq \left|p_0+E_1(\lambda_0)+E_2(\lambda_0)\right|-CT^2\epsilon^2_1\,.
\end{aligned}
\end{equation}
where we use $T>\frac{C}{\Delta}\log^{1/2}\left(\frac{1}{\epsilon_1}\right)$ and $\left|E_3(\lambda_0)\right|\leq CT^2\epsilon^2_1$ in the inequality.

Because $f(\theta^*)>f(\lambda_0)$, we must have
\[
\begin{aligned}
\left|p_0+E_1(\lambda_0)+E_2(\lambda_0)\right|-\left|p_0\exp\left(-\frac{T^2(\theta^*-\lambda_0)^2}{2}\right)+E_1(\lambda_0)+E_2(\lambda_0)\right|\leq CT^2\epsilon^2_1.
\end{aligned}
\]
Because $|E_1(\lambda_0)|+|E_2(\lambda_0)|\leq \frac{T^2\epsilon^2_0}{4}\leq p_0\exp\left(-\frac{T^2\epsilon^2_0}{2}\right)$ and $p_0>0.5$, choosing the constant $C$ properly, the above inequality implies
\[
\exp\left(-\frac{T^2(\theta^*-\lambda_0)^2}{2}\right)\geq 1-\frac{3T^2\epsilon^2_1}{2}\,,
\]
which implies \cref{eqn:condition_lem_2}.
\end{proof}
Using Lemma \ref{lem_new_loss_improve}, we directly have the iteration lemma:
\begin{lem}\label{lem_new_loss_iteration} Given a decreasing positive sequence $\{\epsilon_k\}^K_{k=1}$ with $0<\epsilon_0\leq 1/2$ and $0<\eta<1$, we assume $T,\gamma,\delta,N$ satisfy \cref{eqn:condition_T_prop} and \cref{eqn:condition_N_prop} with $\epsilon_0$. Then we can choose
\begin{equation}\label{eqn:condition_T_lem_it}
T=\Omega\left(\frac{1}{\Delta}\log^{1/2}\left(\frac{1}{\epsilon_K}\right)\right),\ \gamma=\Omega\left(\log^{1/2}\left(\frac{1}{T\epsilon_K}\right)\right),\ \delta=\mathcal{O}\left(\min_{0\leq k\leq K-1}\left\{\frac{\epsilon^2_{k+1}}{\gamma \epsilon_k}\right\}\right)\,,
\end{equation}
and
\begin{equation}\label{eqn:condition_N_lem_it}
N_t=\Omega\left(\max_{0\leq k\leq K-1}\left\{\frac{\gamma^2\epsilon^2_k}{T^2\epsilon^4_{k+1}}\left[\log\left(\frac{K\gamma }{\eta\epsilon_{k+1}}\right)+\alpha\gamma T\right]\exp(2\alpha\gamma T)\right\}\right),\ N_{s,2}=\Theta(1)
\end{equation}
to guarantee that with probability $1-\eta$, 
\begin{equation}\label{eqn:condition_lem_it}
|\theta^*-\lambda_0|\leq \epsilon_K\,.
\end{equation}
\end{lem}
Now, we are ready to prove \cref{thm_alg_new}:
\begin{proof}[Proof of \cref{thm_alg_new}]

Given $\epsilon>0$ small enough, we construct a decreasing positive sequence $\{\epsilon_k\}^K_{k=1}$ with
\[
\epsilon_k=\epsilon^{\frac{2-(1/2)^k}{2-(1/2)^K}}\,.
\]
With this choice, we have
\[
\frac{\epsilon^2_{k+1}}{\epsilon_k}=\epsilon^{\frac{2}{2-(1/2)^K}}
\]
Setting $o=(1/2)^{K-1}$ and noticing $\exp(2\alpha \gamma T)=\left(\frac{1}{\epsilon}\right)^{\Theta\left(\frac{\alpha}{\gamma}\right)}$ when $T=\Theta\left(\frac{1}{\Delta}\log^{1/2}\left(\frac{1}{\epsilon}\right)\right),\ \gamma=\Theta\left(\log^{1/2}\left(\frac{1}{\epsilon}\right)\right)$, Theorem \ref{thm_alg_new} is then a direct result of Lemma \ref{lem_new_loss_iteration}.
\end{proof}

\section{Numerical simulation details for Section \ref{sec:ins}}\label{sec:implemention_details_ideal}

In this section, we provide a detailed description of the classical simulation of RPE, QCELS, and QPE under global depolarizing noise:
\begin{itemize}
\item (RPE): Given a $T_{\max}$, RPE approximates $\lambda_0$ by
\[
\theta^*=\frac{-1}{T_{\max}}\mathrm{atan2}\left(q,p\right)
\]
where $p=\frac{1}{N_{RPE}}\sum^{N_{RPE}}_{n=1}X_n$, $q=\frac{1}{N_{RPE}}\sum^{N_{RPE}}_{n=1}Y_n$, 
and $\{(X_n,Y_n)\}^{N_{RPE}}_{n=1}$ are generated using the Hadamard test circuit (in \cref{fig:flowchart}) same as \eqref{eqn:X} and \eqref{eqn:Y}, which contains the global depolarizing noise in the noisy channel $\mathcal{M}_t$. In our simulation, we set $N_{\mathrm{RPE}}=10^6$ times for all $T_{\max}$.

For simplicity, after calculating $\theta^*$, we directly set the final approximation
\[
\widetilde{\lambda}_0=\mathrm{argmin}_{\theta=\theta^*+\frac{2\pi k}{T_{\max}}}\left\{\left|\theta-\lambda_0\right|\right\}\,.
\]
It is important to note that, in a real application, we need to implement RPE iteratively with an increasing sequence of $T_{\max}$ to avoid the aliasing problem~\cite{NiLiYing2023low}. 

    \item (QCELS): The dataset used in QCELS is the same as \cref{alg:main}. With given values of $N_t$ and $N_{s,2}$, we generate a dataset $\left\{(t_{n},Z_{n})\right\}^{N_t}_{n=1}$ using Algorithm \ref{alg:data}, which utilizes the same probability density $a(t)$ and produces $N_{s,2}$ data samples for each $t_n$. We generalize QCELS by incorporating the contribution from the depolarizing channel: 
    \begin{equation}\label{eqn:QCELS}
\widetilde{\lambda}_0=\theta^*_2=\mathrm{argmin}_{r\in\mathbb{C},\theta_1\in\mathbb{R},\theta_2\in[-\pi,\pi]}\frac{1}{N_t}\sum^{N_t}_{n=1}\left|Z_n-r\exp(-\theta_1|t_n|)\exp(-i\theta_2 t_n)\right|^2\,.
    \end{equation}
We would like to emphasize that the loss function in \eqref{eqn:QCELS} is different from the one used in the original paper \cite{DingLin2023}. Here, the term $\exp(-\theta_1|t_n|)$ is added to fit the global depolarizing noise. However, despite these modifications, QCELS still struggles to efficiently estimate the ground-state energy in the presence of global depolarizing noise. Our analysis in Appendix \ref{sec:analysis_QCELS} demonstrates that regardless of the magnitude of $T_{\max}$, the error of QCELS remains bounded below by a positive constant determined solely by $\alpha$. 

    \item (QPE): Before introducing how to simulate QPE classically under global depolarizing noise, we first provide a brief review of the sample distribution of classical QPE without depolarizing noise. The quantum process of QPE involves executing a series of controlled time evolution operations $e^{-iH}$ on a state $\ket{0^d}|\psi\rangle=\sum^{M-1}_{m=0} c_{m}\ket{0^d}\left|\psi_{m}\right\rangle$. Here,  $\left|\psi_{m}\right\rangle$ denotes eigenstates associated with eigenvalues $\lambda_{m}$. The quantum state resulting from these operations, prior to the application of the inverse Quantum Fourier Transform (QFT), is as follows:
\[
|\Psi\rangle=\frac{1}{\sqrt{N_t}} \sum_{j=-N_t / 2}^{N_t / 2-1}|j\rangle e^{-i j H}|\psi\rangle\,,
\]
where $N_t=2^d$. In the ideal error-free situation, after applying the inverse QFT and measuring the ancilla register, we will get $k$ with probability
\begin{equation}\label{eqn:PK}
P(k)=\sum^{M-1}_{m=0}\left|c_{n}\right|^{2} K_{N_t}\left(\frac{2 \pi k}{N_t}- \lambda_{m}\right),
\end{equation}
where $-N_t/2\leq k\leq N_t/2-1$ and $K_{N_t}$ is the squared and normalized Dirichlet kernel $K_{N_t}(\theta)=\frac{\sin ^{2}(\theta N_t / 2)}{N_t^{2} \sin ^{2}(\theta / 2)}$. To simulate QPE classically, we sample this distribution $N_{QPE}$ times to obtain a set of samples $\{k_i\}^{N_{QPE}}_{i=1}$. We can then approximate the ground state energy as $\widetilde{\lambda}_0=\frac{2\pi\min_{i}k_i}{N_t}$.

Now we analyze the effect of the global depolarizing noise. In this setting, each time the controlled time evolution $e^{-iH}$ is applied, a global depolarizing noise with strength $e^{-\alpha}$ is introduced. Since depolarizing noise commutes with all quantum channels, we can move it to the end of the computation, resulting in a depolarizing noise channel with strength $e^{-\alpha N_t / 2}$. Therefore, the quantum state we obtain before applying the inverse QFT is as follows:
\[
e^{-\alpha N_t / 2}|\Psi\rangle\langle\Psi|+\left(1-e^{-\alpha 
N_t / 2}\right) \frac{I}{M} .
\]
Applying inverse QFT, and measuring the ancilla register will yield $m$ with a probability
\begin{equation}\label{eqn:PK_noise}
P^{\prime}(k)=e^{-\alpha N_t / 2} P(k)+\frac{1-e^{-\alpha N_t / 2}}{N_t} .
\end{equation}
where $P(k)$ is defined in \eqref{eqn:PK}. Similar to before, in the classical simulation, we sample this distribution $N_{QPE}=15$ times and approximate the ground state energy as $\widetilde{\lambda}_0=\frac{2\pi\min_{i}k_i}{N_t}$.
\end{itemize}

\section{Numerical simulation details for Section \ref{sec:numerical_robust}}\label{sec:implemention_details_robust}
\subsection{Implementation of the controlled time evolution matrix}
We conduct numerical tests on our algorithm to estimate the ground-state energy of the transverse field Ising model (TFIM) defined in \cref{eqn:H_Ising}. In practice, the time evolution is implemented by discretizing time into small time steps $\tau$ using the Trotter splitting algorithm. In our numerical experiments, we set the Trotter time step to $\tau = 0.01$. However, due to the discretized evolution time, the forward and backward time evolutions in the quantum benchmarking circuits are not perfectly canceled:
\begin{equation*}
    U_{- t_\mathrm{backward}} U_{t_\mathrm{forward}} = \ket{0}\bra{0} \otimes I_M + \ket{1}\bra{1} \otimes \exp\left(- \I (t_\mathrm{forward} - t_\mathrm{backward}) H \right).
\end{equation*}
Here, $t_\mathrm{forward} \in \tau \ZZ$ is chosen so that $t_\mathrm{forward} + t_\mathrm{backward} = t$ and $\abs{t_\mathrm{forward} - t_\mathrm{backward}} < \tau$. Consequently, the measurement outcome of the quantum benchmarking circuit admits the following bound:
\begin{equation*}
    0 \le \frac{\exp(-\alpha\abs{t}) - \mathbb{E}\left(B_{t, n}\right)}{\exp(-\alpha\abs{t})} \le \frac{\tau^2}{2} \norm{H}_2^2.
\end{equation*}
Hence, mitigating the relative error of the estimation of noise strength using quantum benchmarking circuits requires that the time step $\tau$ should not be too large. 

Computation using a quantum computer relies on a basic gate set, where complex logical gate operators are synthesized from gates within this gate set. In current quantum computer architectures, a universal gate set consisting of one- and two-qubit gates is commonly chosen as the basis gate set. However, directly controlling the time evolution of the Hamiltonian requires many quantum gates to control multiple qubits simultaneously. As a result, the physical implementation of directly controlled Hamiltonian evolution becomes costly due to the need for synthesizing multi-qubit control gates.

Nevertheless, a concept known as control-free implementation, proposed in Ref. \cite{DongLinTong2022ground}, offers a solution to this challenge by exploiting the anti-commutation relation of the Hamiltonian or its components. Instead of directly controlling the Hamiltonian, controlled Hamiltonian evolution can be achieved by controlling other simple gates that obey the anti-commutation relation. In our numerical test, it is noteworthy that the Pauli string $K = Y_1 \otimes Z_2 \otimes Y_3 \otimes Z_4 \otimes \cdots$ anti-commutes with the TFIM Hamiltonian, yielding $K H K = - H$. Therefore, it is sufficient to control this Pauli string in order to implement controlled time evolution, as depicted in \cref{fig:TFIM_circuits}.

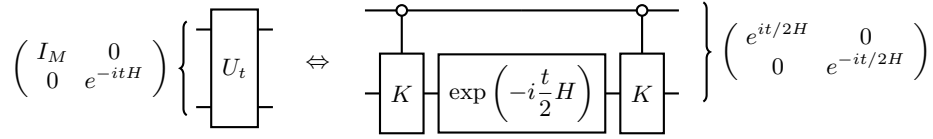
\begin{figure}[htbp]
    \centering
\begin{quantikz}[column sep=0.2cm]
     \lstick[2]{$\left(\begin{array}{cc}
        I_M & 0 \\
        0 & e^{- \I t H}
    \end{array}\right)$} & \gate[2]{U_t} & \qw \\
     & & \qw
  \end{quantikz} $\quad\Leftrightarrow$ \begin{quantikz}[column sep=0.2cm]
    \lstick{} & \octrl{1} & \qw & \octrl{1} & \qw & \rstick[2]{$\left(\begin{array}{cc}
        e^{\I t / 2 H} & 0 \\
        0 & e^{-\I t / 2 H}
    \end{array}\right)$}\\
     & \gate[style={inner ysep=8pt}]{K} & \gate{\displaystyle \exp\left(- \I \frac{t}{2} H\right)} & \gate[style={inner ysep=8pt}]{K} & \qw & 
  \end{quantikz}
    \caption{Controlled time evolution of the TFIM Hamiltonian without directly controlling the Hamiltonian. These quantum circuits are equivalent in the sense that the measurement probabilities of Hadamard tests are the same.\label{fig:TFIM_circuits}}
\end{figure}

\subsection{Noise models and the computation of fidelity parameters}
To emulate the noisy quantum device, we perform noisy quantum simulation using \textsf{IBM Qiskit}. The numerical tests are run with the multiple types of quantum noises which are listed in \cref{tab:noise}. 

\begin{table}[htbp]\renewcommand{\arraystretch}{2}
    \centering
        \begin{tabular}{p{4cm}|>{\centering\arraybackslash}p{5cm}|>{\centering\arraybackslash}p{5cm}}
        \hline
        \hline
                    & Single-qubit gate noise  & Two-qubit gate noise\\
        \hline
        Depolarizing noise & \multicolumn{2}{c}{$\displaystyle \mc{E}_{n, \eta}(\rho) = \eta \rho + \frac{1 - \eta}{4^n - 1} \sum_{P \in \mathsf{Pauli}(n) \backslash \{I_{2^n}\}} P \rho P^\dagger$} \\
        \hline
        Phase flip noise & $\displaystyle \mc{E}_{1, \eta}(\rho) = \eta \rho + (1 - \eta) Z \rho Z$ & $\displaystyle \mc{E}_{2, \eta}(\rho) = \mc{E}_{1, \sqrt{\eta}} \otimes \mc{E}_{1, \sqrt{\eta}}$ \\
        \hline
        Bit flip noise & $\displaystyle \mc{E}_{1, \eta}(\rho) = \eta \rho + (1 - \eta) X \rho X$ & $\displaystyle \mc{E}_{2, \eta}(\rho) = \mc{E}_{1, \sqrt{\eta}} \otimes \mc{E}_{1, \sqrt{\eta}}$\\
        \hline
        General Pauli noise & \multicolumn{2}{c}{$\displaystyle \mc{E}_{n, \eta}(\rho) = \eta \rho + (1 - \eta) \sum_{P \in \mathsf{Pauli}(n) \backslash \{I_{2^n}\}} \gamma_P P \rho P^\dagger$} \\
        \hline 
        \hline
        \end{tabular}
    \caption{Noisy quantum channels used in numerical tests.}
    \label{tab:noise}
\end{table}

To ensure a fair comparison between the numerical results obtained under different noise models, we determine the fidelity parameters for each noise model using the following procedure. We assume that the fidelity parameter $\eta_n$ remains the same for any $n$-qubit noise channel. As a result, the overall composite quantum channel can be expressed as
\begin{equation}
    \mc{E}(\rho_\mathrm{in}) = \eta_\mathrm{tot} \rho_\mathrm{exact} + (1 - \eta_\mathrm{tot}) \sigma_\mathrm{noise}.
\end{equation}
Here, the first component represents the density matrix generated by applying the exact quantum circuit, while the second component contains density matrices that include at least one unwanted operator due to the presence of noise. Although the second component may contain some exact results due to potential cancellation, the density matrix is typically considered as a garbage state. When the quantum circuit consists of $n_{g,1}$ single-qubit gates and $n_{g,2}$ two-qubit gates, the overall fidelity can be calculated as
\begin{equation}
    \eta_\mathrm{tot} = \eta_1^{n_{g,1}} \eta_2^{n_{g,2}}.
\end{equation}
We assume that the two-qubit gate is noisier than the single-qubit gate, such that the infidelities of the noise channels satisfy $1 - \eta_2 = 10(1 - \eta_1)$. Conversely, when a noise strength parameter $\alpha$ is provided, these fidelity parameters can be determined by solving the following nonlinear equation
\begin{equation}\label{eqn:noise_strength_to_params}
    e^{-\alpha |t|} = \eta_\mathrm{tot} = \eta_1^{n_{g,1}} \eta_2^{n_{g,2}}, \ \text{subject to } 1 - \eta_2 = 10(1 - \eta_1).
\end{equation}
To emulate coherent unitary noise, each gate is multiplied by a unitary operator representing the noise. For single-qubit gates acting on the $j$-th qubit, the additional noise unitary is given by $e^{- i \gamma_1 X_j} = \cos(\gamma_1) I - i \sin(\gamma_1) X_j$. For two-qubit gates acting on the $(j, k)$-th qubits, the additional noise unitary is represented by $e^{- i \gamma_2 Z_jZ_k} = \cos(\gamma_2) I - i \sin(\gamma_2) Z_j Z_k$. By decomposing these noise unitary matrices, the quantum channel can be expressed as
\begin{equation}
    \mc{E}(\rho_\mathrm{in}) = (\cos^2 \gamma_1)^{n_{g, 1}} (\cos^2 \gamma_2)^{n_{g, 2}} \rho_\mathrm{exact} + (1 - (\cos^2 \gamma_1)^{n_{g, 1}} (\cos^2 \gamma_2)^{n_{g, 2}}) \sigma_\mathrm{noise}.
\end{equation}
Therefore, by defining 
\begin{equation}\label{eqn:coherent_noise_strength}
    \eta_1 = \cos^2 \gamma_1 \text{ and } \eta_2 = \cos^2 \gamma_2,
\end{equation}
the derivation is identical to the case of general quantum noise channels in \cref{eqn:noise_strength_to_params}. 

\subsection{Choice of parameters}

In numerical simulations for \cref{alg:main} and QCELS, we set the number of time samples to $N_t=10,000$, the number of measurement samples to $N_{s,2} = 500$. 
We would like to emphasize that, as per Theorem \ref{thm_alg_new_informal}, selecting $N_{s,2}=\mathcal{O}(1)$ is adequate for efficiently estimating $\lambda_0$. However, our experimental findings indicate that employing a relatively large value of $N_{s,2}$ can enhance the method's stability in the presence of various noise sources.
We  set the maximum number of measurement samples for each quantum benchmarking circuit to $N_{s,1}=10,000$. In order to perform regression and infer the noise strength, we measure $N_B = 10$ quantum benchmarking circuits, each with a distinct evolution-time parameter. 

\section{Additional numerical results}\label{sec:appendix_additional_numerical_results}
\subsection{Energy landscape of optimization}

To assess the performance of \cref{alg:main}, we analyze the optimization landscape for each parameter set, as depicted in \cref{fig:tfim_landscape_multiple_T}. The visualizations reveal that increasing $T_\text{max}$ leads to a more concentrated optimization landscape centered around the global minimum, which closely approximates the true ground-state energy. Additionally, despite the distinct behavior exhibited by different noise types, the optimization landscapes consistently preserve the global minimum of the ideal loss function defined in \cref{eqn:loss_multi_modal_perfect}. This characteristic elucidates that our ground-state energy estimation algorithm yields highly accurate results, even in the presence of various noise types. Moreover, we illustrate the ground-state energy estimation obtained from three independent repetitions in \cref{fig:tfim_landscape_multiple_T}. The coherence observed across these multiple runs indicates that our algorithm maintains consistency throughout the process, despite differences in the generation of noisy quantum data for each realization.

\begin{figure}[htbp]
    \centering
    \includegraphics[width = \textwidth]{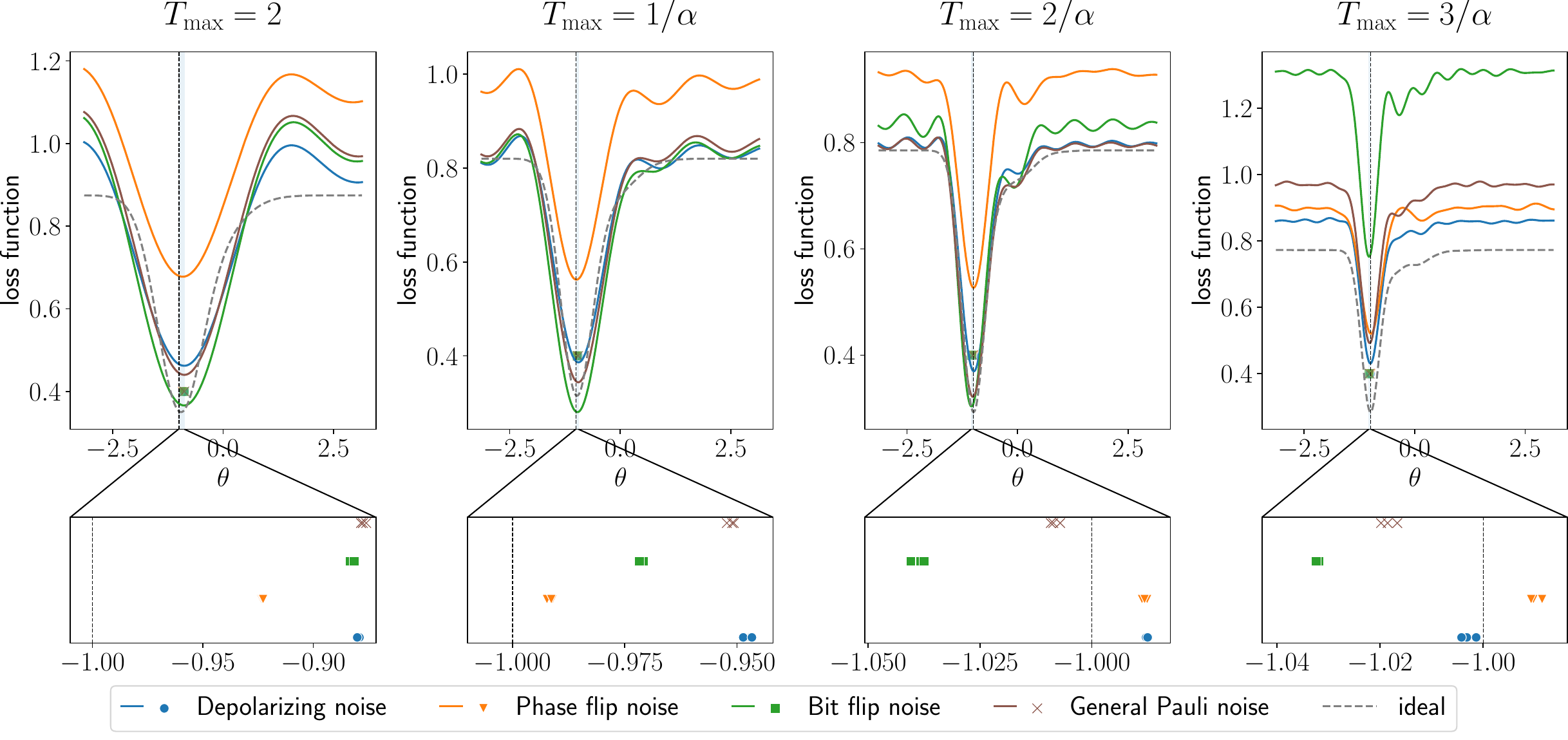}
    \caption{Optimization landscape of the ground-state energy estimation. The noise strength is set to $\alpha = 0.125$. Vertical dashed lines stand for the normalized ground-state energy $-1$. Top panels: Optimization landscape of multiple choices of $T_\text{max}$ and noise. For each set of parameters, we pick one dataset $\mc{D}_H = \{ (t_n, Z_n) \}_{n = 1}^{N_t}$ and one optimization result $(r^*, \theta^*)$ of \cref{alg:main} from three repetitions. The optimization landscape is computed with respect to \cref{eqn:fixed_beta_new} by fixing $r^*$ and $\mc{D}_H$. The ideal optimization landscape is derived from \cref{eqn:loss_multi_modal_perfect} with fixed $r^*$. Bottom panels: Zoomed-in view around the optimum. The scatters are ground-state energy estimations of three repetitions of each set of parameters. }
    \label{fig:tfim_landscape_multiple_T}
\end{figure}

\subsection{Analysis of quantum benchmarking circuits}

Following the steps outlined in \cref{alg:estimation_alpha}, we estimate the noise strength $\alpha$ by conducting a linear regression analysis on the dataset ${(\abs{t_n}, \ln(\bar{B}_n))}_{n = 1}^{N_b}$. The scatter plot of the dataset and corresponding regression lines are presented in \cref{fig:tfim_benchmark}. The numerical results demonstrate a strong fit of the linear model to the dataset, with the exception of the data affected by coherent unitary noise. This observation is further discussed and justified in \cref{sec:justify:coherent_unitary_noise}.

\begin{figure}[htbp]
    \centering
    \includegraphics[width = \textwidth]{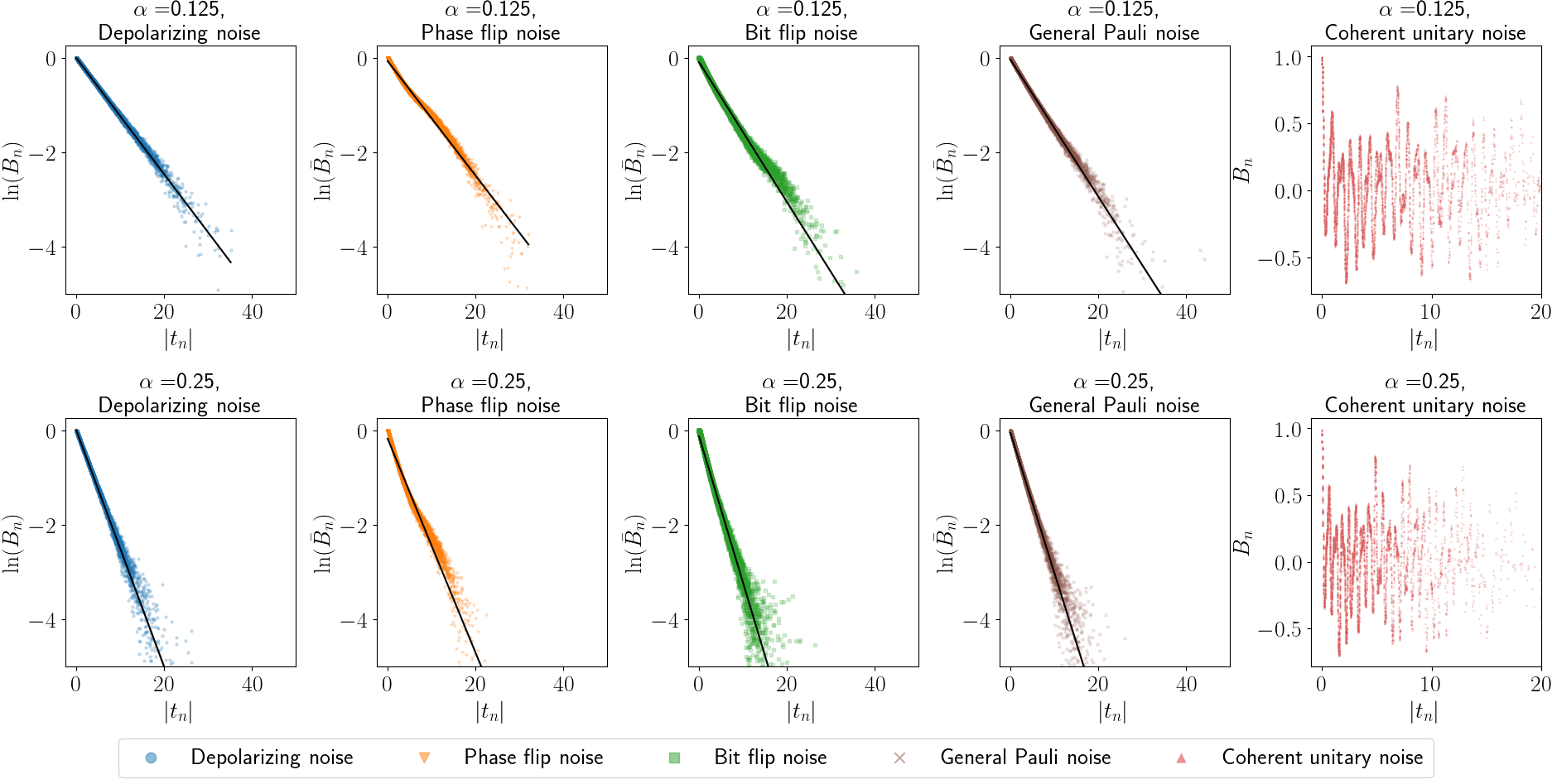}
    \caption{Data generated from quantum benchmarking circuits and its regression.}
    \label{fig:tfim_benchmark}
\end{figure}

In \cref{fig:tfim_reg_size}, we explore how the ground-state energy estimation error depends on the number of quantum benchmarking circuits used, which corresponds to the number of data points utilized for regressing the noise strength parameter. The numerical findings suggest that the dependence is not significant, and even a small number of quantum benchmarking circuits ensures that the ground-state energy estimation procedure outlined in \cref{alg:main} operates robustly.

\begin{figure}[htbp]
    \centering
    \includegraphics[width = \textwidth]{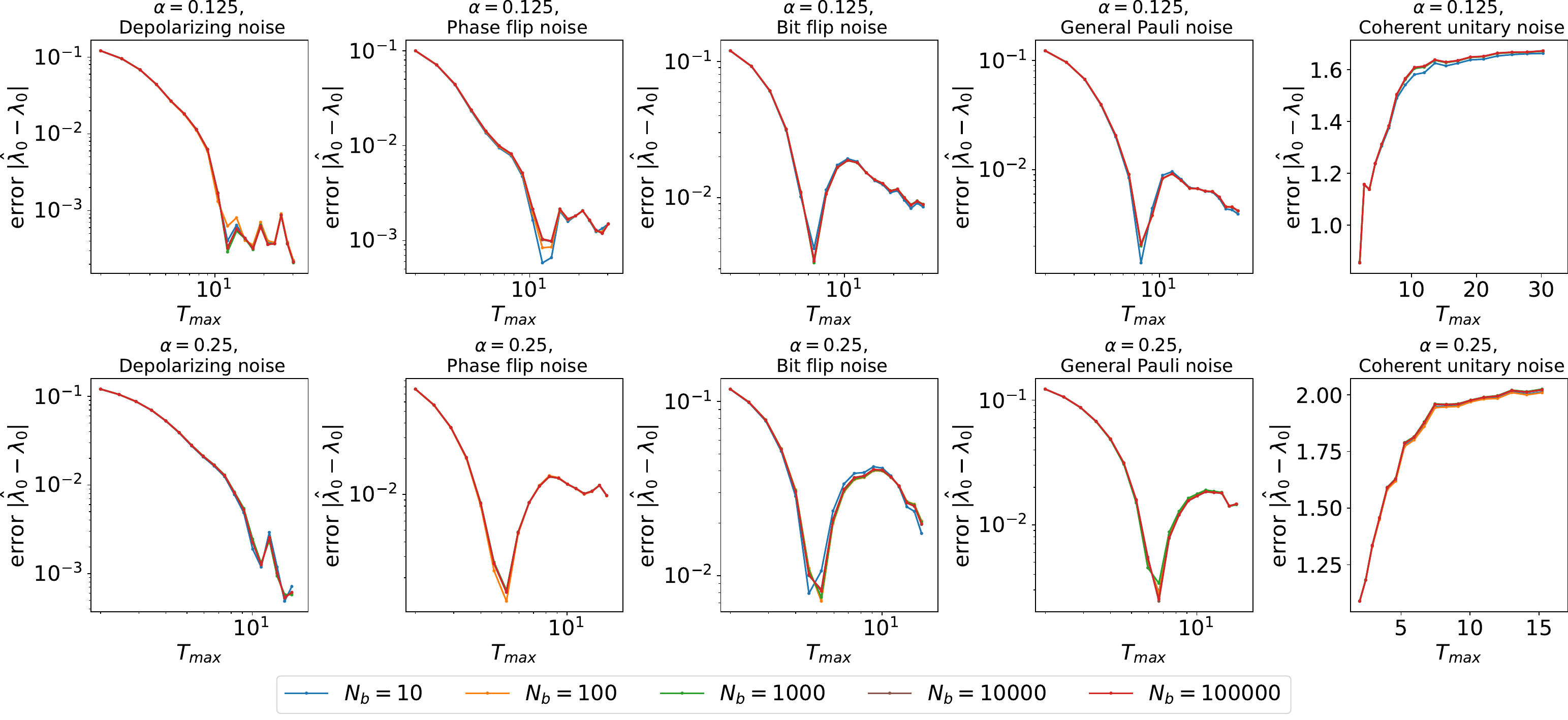}
    \caption{Accuracy of ground-state energy estimation with varying numbers of quantum benchmarking circuits.}
    \label{fig:tfim_reg_size}
\end{figure}

The observed results in \cref{fig:tfim_reg_size} are quite surprising. According to standard regression analysis, if the number of quantum benchmarking circuits, denoted as $N_b$, is decreased, the accuracy of noise-strength estimation is expected to diminish as $\Or(1/\sqrt{N_b N_s})$. However, \cref{fig:tfim_reg_size} reveals that even with less accurate noise-strength estimation, \cref{alg:main} is still capable of producing accurate ground-state energy estimations. To comprehend the algorithm's robustness, we visualize the optimization landscape under different noise-strength estimations using various numbers of quantum benchmarking circuits, as shown in \cref{fig:tfim_landscape}. Several observations can be made from these visualizations.

Firstly, as $N_b$ increases and the noise strength is estimated more accurately, the landscape curves converge. However, the limit curve does not precisely match the ideal landscape, indicating that the overall noise does not perfectly align with the ideal noise ansatz in all cases.

Secondly, although the landscape curves differ across cases, they retain the most crucial characteristic by attaining the global minimum at the exact ground-state energy. Therefore, even with a less accurate noise-strength estimation as the input, \cref{alg:main} can effectively identify the ground-state energy. This feature contributes to the algorithm's robustness in ground-state energy estimation.

Notably, the optimization landscape in the presence of coherent unitary noise significantly deviates from those of other quantum noises. This discrepancy arises due to the ineffectiveness of the quantum benchmarking circuit in mitigating coherent unitary noise. Furthermore, as discussed in \cref{sec:justify:coherent_unitary_noise}, coherent unitary noise introduces an intrinsic ground-state energy estimation error into the data. However, applying Pauli twirling to the quantum circuit can transform coherent unitary noise into Pauli noises, thereby enhancing the algorithm's performance by rectifying the ineffectiveness caused by coherent unitary noise.

\begin{figure}[htbp]
    \centering
    \includegraphics[width = \textwidth]{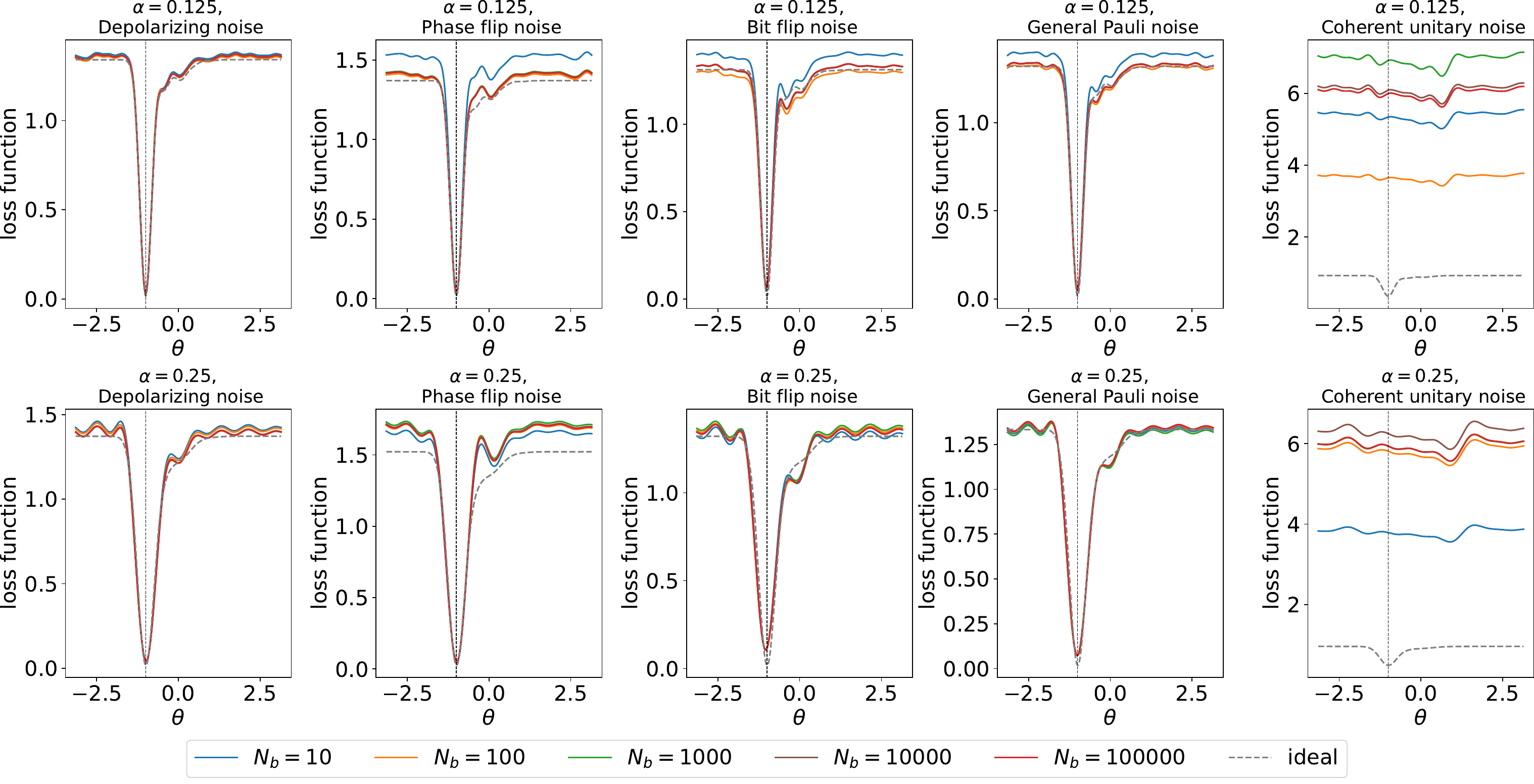}
    \caption{Optimization landscape of ground-state energy estimation with different numbers of quantum benchmarking circuits. Vertical dashed lines represent the normalized ground-state energy of $-1$.}
    \label{fig:tfim_landscape}
\end{figure}

\section{Justification of the ineffectiveness against coherent unitary noise}\label{sec:justify:coherent_unitary_noise}
The numerical results have indicated that our main algorithm is ineffective when applied to data affected by coherent unitary noise. However, by applying Pauli twirling and transforming the coherent unitary noise to Pauli noise, the algorithm's performance can be significantly improved. Understanding the reasons behind this ineffectiveness is an intriguing question. In this section, we will analyze the impact of coherent unitary noise and provide a justification for its ineffectiveness.

In our numerical simulations, the coherent unitary noise on each single-qubit gate is represented by a single-qubit rotation unitary $\exp(- \I \gamma_1 X_j)$, while on each two-qubit gate, it is represented by a two-qubit rotation unitary $\exp(- \I \gamma_2 Z_jZ_k)$. It is important to note that the dominant components in the quantum circuit are those associated with Trotterized Hamiltonian simulation, which utilize $\exp(\I g \tau X_j)$, $\exp(\I g \tau / 2 X_j)$, and $\exp(\I \tau Z_j Z_k)$ gates. In the presence of coherent unitary noise, these gates are modified as 
\begin{equation}
    \begin{split}
        & \exp(\I g \tau X_j) \to \exp(\I (g \tau - \gamma_1) X_j), \quad \exp(\I g \tau / 2 X_j) \to \exp(\I (g \tau / 2 - \gamma_1) X_j),\\
        & \exp(\I \tau Z_j Z_k) \to \exp(\I (\tau - \gamma_2) Z_j Z_k).
    \end{split}
\end{equation}
Indeed, the impact of coherent unitary noise on the gate implementation introduces a bias in the angle parameters. Consequently, the Trotter part of the quantum circuit effectively approximates the evolution of the following Hamiltonian:
\begin{equation}
    \exp\left( - \I r \tau \left(- g\left(1 - \frac{\gamma_1}{g \tau} \right) \sum_{j = 1}^L X_j - \left( 1 - \frac{\gamma_2}{\tau} \right) \sum_{j = 1}^{L - 1} Z_j Z_{j + 1}\right)\right)
\end{equation}
where $r = t / \tau$ is the number of Trotter steps. For simplicity, we assume that $\gamma_1 = \gamma_2 = \gamma$ and $g = 1$. In this case, the effective Hamiltonian becomes a constant multiple of the original Hamiltonian
\begin{equation}\label{eqn:effective_Hamiltonian}
    H_\text{eff} = (1 - \gamma / \tau) H.
\end{equation}
The noise parameters is connected with the noise strength $\alpha$ by \cref{eqn:noise_strength_to_params,eqn:coherent_noise_strength} which gives
\begin{equation}
    \exp(- \alpha t) = \cos^{2n_g}(\gamma) \Rightarrow \alpha t = - 2 n_g \ln \cos(\gamma) = n_g \gamma^2 + \Or(n_g \gamma^4).
\end{equation}
The total number of gates in the Trotter part is $n_g = r(L-1) + (r + 1) L \approx 2 r L$ when using the second order Trotter scheme. Thus, the following relation holds
\begin{equation}
    \gamma / \tau = \sqrt{\alpha / (2 \tau L)} ( 1 + \Or(\gamma^2) ).
\end{equation}
Due to the fact that the number of gates in the quantum circuit, apart from the Trotter part, remains constant and does not increase with the simulation time $t$, we can assume that the algorithm is resilient to these consistently small noises. Under this assumption, the algorithm estimates the ground-state energy of the effective Hamiltonian described in \cref{eqn:effective_Hamiltonian}. The error with respect to the ground-state energy of the exact Hamiltonian can then be expressed as follows
\begin{equation}
    \frac{\abs{\lambda_\text{min}(H_\text{eff}) - \lambda_\text{min}(H)}}{\abs{\lambda_\text{min}(H)}} = \frac{\gamma}{\tau} \approx \sqrt{\frac{\alpha}{2 \tau L}}.
\end{equation}
This equality reveals that, for a given fixed setup, the ground-state energy estimation in the presence of coherent unitary noise is inherently subject to an error of magnitude $\Or(\sqrt{\alpha})$. This occurs because the algorithm is unable to distinguish the changes caused by the coherent unitary noise in the effective Hamiltonian.

Furthermore, when applying the quantum benchmarking circuit in the presence of coherent unitary noise, the noise term does not cancel out during the forward and backward time evolutions. As a result, the Trotter part implements $\exp(\I \gamma / \tau H)$ instead of an identity matrix. Neglecting the errors from the Trotter scheme and the noise on the gates other than those in the Trotter part, the difference between the probability of measuring 0 and the probability of measuring 1 is given by
\begin{equation}
    B(t) = \braket{\psi | \cos(t \gamma / \tau H) | \psi} = \sum_{j} p_j \cos(t \gamma \lambda_j / \tau).
\end{equation}
Consequently, the output of the quantum benchmarking circuit exhibits an oscillatory function with respect to the time parameter $t$, rather than an exponential decaying factor. This significant deviation from the original ansatz for deriving the noise strength $\alpha$ renders the quantum benchmarking circuit ineffective in this scenario. These analyses align with the numerical results presented in \cref{fig:tfim_benchmark}.

\section{Complexity of robust phase estimation
(RPE)}\label{sec:intuitive_hardness}

In this section, we consider using robust phase estimation (RPE)~\cite{PhysRevA.104.069901, NiLiYing2023low} to deal with depolarizing noise. Recall that: Given a fixed $T>0$, RPE approximates $\lambda_0$ by
\begin{equation}\label{eqn:theta_star}
\theta^*=-\frac{1}{T}\mathrm{atan2}\left(q,p\right)
\end{equation}
where $\mathrm{atan2}\left(x,y\right)\in(-\pi,\pi]$ is defined as the angle of the complex number $x+iy$. Here $p=\frac{1}{N}\sum^N_{n=1}X_n$, $q=\frac{1}{N}\sum^N_{n=1}Y_n$, 
and $\{(X_n,Y_n)\}^{N}_{n=1}$ are generated using the quantum circuit in \cref{fig:flowchart} same as \eqref{eqn:X} and \eqref{eqn:Y}.

For simplicity, we assume $p_0=1$, and $p_k=0$ for $k>0$. Recall \eqref{eqn:E_n}, we can write $p,q$ as
\[
p=\exp(-\alpha T)\cos(-\lambda_0T)+\frac{1}{N}\sum^N_{n=1}E_{n,1},\quad q=\exp(-\alpha T)\sin(-\lambda_0T)+\frac{1}{N}\sum^N_{n=1}E_{n,2}\,.
\]
Here $\left\{E_{n,1},E_{n,2}\right\}^N_{n=1}$ are independent and bounded complex random variables with zero expectation, $\mathrm{Var}(E_{n,1})=\Omega(1)$, and $\mathrm{Var}(E_{n,2})=\Omega(1)$. 

First, we notice that
\[
\mathrm{atan2}\left(p,q\right)=\mathrm{atan2}\left(\exp(\alpha T)p,\exp(\alpha T)q\right)=T\lambda_0\,.
\]
This implies that when there is no random noise ($E_{n,1}=E_{n,2}=0$), the result of RPE is not affected by the global depolarizing noise.
 
Next, we consider the case with noise. For fixed $T>0$, in order to achieve $|\theta^*-\lambda_0|\leq \epsilon$, it is necessary to satisfy the condition:
\[
\left|\mathrm{atan2}\left(q,p\right)+T\lambda_0\right|\leq T\epsilon\Leftrightarrow \left|\mathrm{atan2}\left(\exp(\alpha T)q,\exp(\alpha T)p\right)+T\lambda_0\right|\leq T\epsilon\,.
\]
Consequently, to reach $\epsilon$ accuracy, we need to choose a sufficiently large value of $N$ so that:
\[
\left|\frac{\exp(\alpha T)}{N}\sum^N_{n=1}E_{n,1}\right|=\left|\frac{\exp(\alpha T)}{N}\sum^N_{n=1}E_{n,2}\right|=\mathcal{O}(T\epsilon)\,.
\]
Noticing $\mathrm{Var}(E_{n,1})=\Omega(1)$ and $\mathrm{Var}(E_{n,2})=\Omega(1)$, we need a minimum number of measurements given by 
\[
N=\Omega(\exp(2\alpha T)/(T^2\epsilon^2))\,.
\]
It is important to note that in practical implementations, RPE iteratively implements \eqref{eqn:theta_star} to address the aliasing problem (for further details, refer to~\cite{NiLiYing2023low}). Additionally, when $p_0<1$, to achieve accuracy $\epsilon$, the required maximum Hamiltonian running time $T_{\max}$ should be at least $\Omega\left(\frac{1}{\epsilon}\right)$, which implies that the total running time is at least $NT_{\max}=\Omega\left(\exp\left(\frac{2\alpha}{\epsilon}\right)\right)$.

\section{Analysis of QCELS}\label{sec:analysis_QCELS}

In this section, we provide a simple example to illustrate that under the influence of global depolarizing noise, QCELS is unable to achieve arbitrarily small errors in estimating the ground state energy.

As discussed in \cref{sec:implemention_details_ideal}, the approximation $\wt{\lambda}_0$ from QCELS is constructed by solving the following optimization problem:
\[
\wt{\lambda}_0=\theta^*_2,\quad (r^*,\theta^*_1,\theta^*_2)=\mathrm{argmin}_{r\in\mathbb{C},\theta_1\in\mathbb{R},\theta_2\in[-\pi,\pi]}\frac{1}{N_t}\sum^{N_t}_{n=1}\left|Z_n-r\exp(-\theta_1|t_n|)\exp(-i\theta_2 t_n)\right|^2\,.
\]
Define 
\[
\begin{aligned}
L(r,\theta_1,\theta_2)=&\frac{1}{N_t}\sum^{N_t}_{n=1}\left|Z_n-r\exp(-\theta_1|t_n|)\exp(-i\theta_2 t_n)\right|^2\\
=&\frac{1}{N_t}\sum^{N_t}_{n=1}\exp(-2\theta_1|t_n|)\left|\exp(\theta_1|t_n|+i\theta_2 t_n)Z_n-r\right|^2
\end{aligned}\,.
\]
Since $L$ is a quadratic function for $r$, we obtain
\[
r^*=\frac{1}{C_{\theta_1^*} N_t}\sum^{N_t}_{n=1}Z_n\exp(-\theta^*_1|t_n|+i\theta^*_2 t_n)\,,
\]
where $C_{\theta_1^*}=\frac{1}{N_t}\sum^{N_t}_{n=1}\exp(-2\theta^*_1|t_n|)$.
Given the value of $\theta_1^*$, we can reformulate the maximization problem for $\theta_2^*$. Specifically, plugging this into $L$, we further obtain that 
\[
\begin{aligned}
\theta^*_2&=\mathrm{argmax}_{\theta_2\in[-\pi,\pi]}\left|\frac{1}{N_t}\sum^{N_t}_{n=1}Z_n\exp(-\theta^*_1|t_n|+i\theta_2 t_n)\right|^2\\
&=\mathrm{argmax}_{\theta_2\in[-\pi,\pi]}\left|\sum^{M-1}_{m=0}p_m\int^\infty_{-\infty}a(t)\exp(-(\alpha+\theta^*_1)|t|)\exp(-i(\lambda_m-\theta_2)t)\rd t+E(\theta_2)\right|^2\\
&=\mathrm{argmax}_{\theta_2\in[-\pi,\pi]}\left|\sum^{M-1}_{m=0}p_m(F*G_{\theta^*_1})(\theta_2-\lambda_m)+E(\theta_2)\right|
\end{aligned}\,,
\]
where
\[
\begin{aligned}
E(\theta_2)=&\sum^{M-1}_{m=0}p_m\left(\frac{1}{N_t}\sum^{N_t}_{n=1}\exp(-(\alpha+\theta^*_1)|t_n|)\exp(-i(\lambda_m-\theta_2)t_n)\right.\\
&\left.-\int^\infty_{-\infty}a(t)\exp(-(\alpha+\theta^*_1)|t|)\exp(-i(\lambda_m-\theta_2)t)\rd t\right)\\
&+\frac{1}{N_t}\sum^{N_t}_{n=1}E_n\exp(-\theta^*_1|t_n|+i\theta_2 t_n)
\end{aligned}
\]
with $E_n$ defined in \eqref{eqn:E_n}. Here $G_{\theta^*_1}(x)=\frac{1}{2\pi}\int^{\infty}_{-\infty}\exp(-(\alpha+\theta^*_1)|t|)\exp(ixt)\rd t=\frac{\alpha+\theta^*_1}{\pi\sqrt{(\alpha+\theta^*_1)^2+x^2}}$ and the error term $E(\theta_2)$  contains the sampling and measuring error and satisfies $|E(\theta_2)|=\mathcal{O}(1/\sqrt{N_t})$.

For simplicity, we assume $\gamma=\infty$, $N_t\gg 1$, and $\theta^*_1\approx \alpha$, then we have $F(x)=\exp(-T^2x^2/2)$ and $E(\theta_2)\approx 0$. This implies
\[
\theta^*_2\approx \mathrm{argmax}_{\theta\in[-\pi,\pi]}\left|\sum^{M-1}_{m=0}p_m(F*G_{\alpha})(\theta-\lambda_m)\right|\,.
\] 
We consider a special case with $p_0>0.5$, $p_1=1-p_0$. Then, the loss function becomes to
\[
H(\theta)=p_0(F*G_{\alpha})(\theta-\lambda_0)+p_1(F*G_{\alpha})(\theta-\lambda_1)
\]
Noticing 
\[
\lim_{T\rightarrow\infty}\frac{T}{\sqrt{2\pi}}H(\theta)=p_0G_{\alpha}(\theta-\lambda_0)+p_1G_{\alpha}(\theta-\lambda_1)\,.
\]
This implies that, even when $T$ is sufficiently large,
\[
\theta^*_2\approx \mathrm{argmax}_{\theta\in[-\pi,\pi]}p_0G_{\alpha}(\theta-\lambda_0)+p_1G_{\alpha}(\theta-\lambda_1)\neq \lambda_0\,.
\]
Thus, QCELS can not achieve arbitrary small errors under global depolarizing noise. This is in contrast to the case where $\alpha=0$, resulting in $G_\alpha=\delta_0(x)$ and $\theta^*_2\approx \lambda_0$ when $T$ is sufficiently large.
\end{document}